%% file: OOB_Correlation_Jnl.tex
\documentclass[12pt, draftclsnofoot,onecolumn]{IEEEtran}

\usepackage{graphicx,subcaption,xcolor,setspace,amsthm,cite,booktabs,epstopdf,dblfloatfix,blindtext,amsthm,hyperref}
\usepackage{algorithm,algcompatible,upgreek}

\usepackage[binary-units,detect-weight=true, detect-family=true]{siunitx}
\DeclareSIUnit \bitspersecond {bps}
\newcommand{\centrallocation}{/Users/AnumAli/Dropbox/STUDY/Central} 
\newtheorem{thm}{Theorem}
\input{macros.tex}

\algnewcommand\INPUT{\item[\textbf{Input:}]}
\algnewcommand\OUTPUT{\item[\textbf{Output:}]}
\algnewcommand\INIT{\item[\textbf{Initialization:}]}
\algnewcommand\OFFLINE{\item[\textbf{Offline Calculation:}]}
\newcommand{\ul}{\underline}
\newcommand{\subsGHz}{sub-6 \GHz}

\newcommand{\SubsGHz}{Sub-6 \GHz}
\newcommand{\GHz}{\SI{}{\giga\hertz}}
\newcommand{\MHz}{\SI{}{\mega\hertz}}

\newcommand{\dBm}{\SI{}{\decibel}\rmm}
\newcommand{\ns}{\SI{}{\nano\second}}
\newcommand{\SNR}{{\rm SNR}}

\newcommand{\bHu}{\ul{\bH}}
\newcommand{\NRX}{N_{\rmR\rmX}}
\newcommand{\NTX}{N_{\rmT\rmX}}
\newcommand{\MRX}{M_{\rmR\rmX}}
\newcommand{\MTX}{M_{\rmT\rmX}}

\newcommand{\ulNRX}{\ul N_{\rmR\rmX}}
\newcommand{\ulNTX}{\ul N_{\rmT\rmX}}
\newcommand{\Ns}{N_{\rms}}
\newcommand{\ulNs}{\ul{N}_{\rms}}

\newcommand{\ulK}{\ul{K}}

\newcommand{\bFBB}{\bF_{\rmB\rmB}}
\newcommand{\bFRF}{\bF_{\rmR\rmF}}
\newcommand{\bWBB}{\bW_{\rmB\rmB}}
\newcommand{\bWRF}{\bW_{\rmR\rmF}}
\newcommand{\bWRFt}{\bW_{\rmR\rmF,t}}
\newcommand{\baTX}{\ba_{\rmT\rmX}}
\newcommand{\baRX}{\ba_{\rmR\rmX}}
\newcommand{\bATX}{\bA_{\rmT\rmX}}
\newcommand{\bARX}{\bA_{\rmR\rmX}}
\newcommand{\bbARX}{\bar{\bA}_{\rmR\rmX}}

\newcommand{\bsfRTX}{\bsfR_{\rmT\rmX}}

\newcommand{\bsfRRX}{\bsfR_{\rmR\rmX}}

\newcommand{\bsfRy}{\bsfR_{\by}}

\newcommand{\hbRyt}{\hat{\bR}_{\by,t}}
\newcommand{\bsfRg}{\bsfR_{\bsfg}}

\newcommand{\bURX}{\bU_{\rmR\rmX}}
\newcommand{\bURXs}{\bU_{\rmR\rmX,\rms}}
\newcommand{\bURXn}{\bU_{\rmR\rmX,\rmn}}

\newcommand{\bUTXs}{\bU_{\rmT\rmX,\rms}}

\newcommand{\hbURX}{\hat{\bU}_{\rmR\rmX}}
\newcommand{\hbURXs}{\hat{\bU}_{\rmR\rmX,\rms}}
\newcommand{\hbURXn}{\hat{\bU}_{\rmR\rmX,\rmn}}

\newcommand{\hbUTXs}{\hat{\bU}_{\rmT\rmX,\rms}}

\newcommand{\bSigmas}{\bSigma_\rms}
\newcommand{\bSigman}{\bSigma_\rmn}
\newcommand{\hbSigma}{\hat{\bSigma}}
\newcommand{\hbSigmas}{\hat{\bSigma}_\rms}
\newcommand{\hbSigman}{\hat{\bSigma}_\rmn}

\newcommand{\hbRbg}{\hat{\bR}_{\bar\bg}}
\newcommand{\hbRbgt}{\hat{\bR}_{\bar\bg,t}}

\newcommand{\delbsfR}{\Delta \bsfR}
\newcommand{\delbsfRRX}{\Delta \bsfR_{\rmR\rmX}}
\newcommand{\delbsfRTX}{\Delta \bsfR_{\rmT\rmX}}
\newcommand{\delURXs}{\Delta \bU_{\rmR\rmX,\rms}}
\newcommand{\delUTXs}{\Delta \bU_{\rmT\rmX,\rms}}

\newcommand{\ulDelta}{\ul\Delta}

\newcommand{\alpharc}{\alpha_{r_c}}
\newcommand{\Ts}{T_\rms}

\begin{document}
\bstctlcite{IEEEmax3beforeetal}
\title{Spatial Covariance Estimation for Millimeter Wave Hybrid Systems using \\Out-of-Band Information}
\author{Anum Ali, {\it Student Member, IEEE}, Nuria Gonz\'alez-Prelcic, {\it Member, IEEE}, and \\Robert W. Heath Jr., {\it Fellow, IEEE}
\thanks{This research was partially supported by the U.S. Department of Transportation through the Data-Supported Transportation Operations and Planning (D-STOP) Tier 1 University Transportation Center and by the Texas Department of Transportation under Project 0-6877 entitled ``Communications and Radar-Supported Transportation Operations and Planning (CAR-STOP)'' and by the National Science Foundation under Grant No. ECCS-1711702. N. Gonz\'alez-Prelcic was supported by the Spanish Government and the European Regional Development Fund (ERDF) under Project MYRADA (TEC2016-75103-C2-2-R).}
\thanks{A. Ali and R. W. Heath Jr. are with the Department of Electrical and Computer Engineering, The University of Texas at Austin, Austin, TX 78712-1687 \mbox{(e-mail: \{anumali,rheath\}@utexas.edu)}.}
\thanks{N. Gonz\'alez-Prelcic is with the Signal Theory and Communications Department, University of Vigo, Vigo, Spain \mbox{(e-mail: nuria@gts.uvigo.es)}.}
\thanks{A preliminary version of this work appeared in the Proceedings of Information Theory and Applications Workshop (ITA), February, 2016~\cite{Ali2016Estimating}.}
}
\maketitle
%
\begin{abstract}
In high mobility applications of millimeter wave (mmWave) communications, e.g., vehicle-to-everything communication and next-generation cellular communication, frequent link configuration can be a source of significant overhead. We use the \subsGHz~channel covariance as an out-of-band side information for mmWave link configuration. Assuming: (i) a fully digital architecture at sub-\SI{6}{\giga\hertz}; and (ii) a hybrid analog-digital architecture at mmWave, we propose an out-of-band covariance translation approach and an out-of-band aided compressed covariance estimation approach. For covariance translation, we estimate the parameters of \subsGHz~covariance and use them in theoretical expressions of covariance matrices to predict the mmWave covariance. For out-of-band aided covariance estimation, we use weighted sparse signal recovery to incorporate out-of-band information in compressed covariance estimation. The out-of-band covariance translation eliminates the in-band training completely, whereas out-of-band aided covariance estimation relies on in-band as well as out-of-band training. We also analyze the loss in the signal-to-noise ratio due to an imperfect estimate of the covariance. The simulation results show that the proposed covariance estimation strategies can reduce the training overhead compared to the in-band only covariance estimation.
\end{abstract}
\section{Introduction}
\IEEEPARstart{T}{he} hybrid precoders/combiners for millimeter wave (mmWave) MIMO systems are typically designed based on either instantaneous channel state information (CSI)~\cite{ElAyach2014Spatially} or statistical CSI~\cite{Alkhateeb2013Hybrid}. Obtaining channel information at mmWave is, however, challenging due to: (i) the large dimension of the arrays used at mmWave, (ii) the hardware constraints (e.g., limited number of RF chains~\cite{ElAyach2014Spatially,Alkhateeb2013Hybrid}, and/or low-resolution analog-to-digital converters (ADCs)~\cite{Mo2015Capacity}), and (iii) low pre-beamforming signal-to-noise ratio (\SNR). We exploit out-of-band information extracted from \subsGHz~channels to configure the mmWave links. The use of \subsGHz~information for mmWave is enticing as mmWave systems will likely be used in conjunction with \subsGHz~systems for multi-band communications and/or to provide wide area control signals~\cite{kishiyama2013future,daniels2007multi,Hashemi2018Out}.

Using out-of-band information can positively impact several applications of mmWave communications. In mmWave cellular~\cite{Pi2011introduction,Bai2015Coverage,2016Federal}, the base-station user-equipment separation can be large (e.g., on cell edges). In such scenarios, link configuration is challenging due to poor pre-beamforming \SNR~and user mobility. The pre-beamforming \SNR~is more favorable at \subsGHz~due to lower bandwidth. Therefore, reliable out-of-band information from \subsGHz~can be used to aid the mmWave link establishment. Similarly, frequent reconfiguration will be required in highly dynamic channels experienced in mmWave vehicular communications~(see e.g.,~\cite{choi2016millimeter,va2016impact} and the references therein). The out-of-band information (coming e.g., from dedicated short-range communication (DSRC) channels~\cite{li2010overview}) can play an important role in unlocking the potential of mmWave vehicular communications.
\subsection{Contributions}
The main contributions of this paper are as follows:
\begin{itemize}
\item We propose an out-of-band covariance translation strategy for MIMO systems. The proposed translation approach is based on a parametric estimation of the mean angle and angle spread (AS) of all clusters at \subsGHz. The estimated parameters are then used in the theoretical expressions of the spatial covariance to complete the translation. 

\item We formulate the problem of covariance estimation for hybrid MIMO systems as a compressed signal recovery problem. To incorporate out-of-band information in the proposed formulation, we introduce the concept of weighted compressed covariance estimation (similar to weighted sparse signal recovery~\cite{Scarlett2013Compressed}). The weights in the proposed approach are chosen based on the out-of-band information. 

\item We use tools from singular vector perturbation theory~\cite{Liu2008First} to quantify the loss in received post-processing \SNR~due to the use of imperfect covariance estimates. Specifically, considering a single path channel, we find an upper and lower bound on the loss in \SNR. The resulting expressions permit a simple and intuitive explanation of the loss in terms of the mismatch between the true and estimated covariance.
\end{itemize}
\subsection{Prior work}
We propose two mmWave covariance estimation strategies. The first strategy is covariance translation from \subsGHz~to mmWave, while the second strategy is out-of-band aided compressed covariance estimation. In this section, we review the prior work relevant to each approach.

Most of the prior work on covariance translation was tailored towards frequency division duplex (FDD) systems~\cite{Aste1998Downlink,Liang2001Downlink,Chalise2004Robust,Hugl1999Downlink,Jordan2009Conversion,Decurninge2015Channel,Miretti2018FDD}. The prior work includes least-squares based~\cite{Aste1998Downlink,Liang2001Downlink,Chalise2004Robust}, minimum variance distortionless response based~\cite{Hugl1999Downlink}, and~\cite{Miretti2018FDD} projection based strategies. In~\cite{Jordan2009Conversion}, a spatio-temporal covariance translation strategy was proposed based on two-dimensional interpolation. In~\cite{Decurninge2015Channel}, a training based covariance translation approach was presented. Unlike~\cite{Aste1998Downlink,Liang2001Downlink,Chalise2004Robust,Hugl1999Downlink,Jordan2009Conversion}, the translation approach in~\cite{Decurninge2015Channel} requires training specifically for translation but does not assume any knowledge of the array geometry. The uplink (UL) information has also been used in estimating the instantaneous downlink (DL) channel~\cite{Vasisht2016Eliminating,Shen2016Compressed}. In~\cite{Vasisht2016Eliminating}, the multi-paths in the UL channel were separated and subsequently used in the estimation of the DL channel. The UL measurements were used to weight the compressed sensing based DL channel estimation in~\cite{Shen2016Compressed}. 

In FDD systems, the number of antennas in the UL and DL array is typically the same, and simple correction for the differences in array response can translate the UL covariance to DL. MmWave systems, however, will use a larger number of antennas in comparison with \subsGHz,~and conventional translation strategies (as in \cite{Aste1998Downlink,Liang2001Downlink,Chalise2004Robust,Hugl1999Downlink,Jordan2009Conversion,Decurninge2015Channel,Miretti2018FDD,Vasisht2016Eliminating,Shen2016Compressed}) are not applicable. Further, the frequency separation between UL and DL is typically small (e.g., there is $9.82\%$ frequency separation between \SI{1935}{\MHz} UL and \SI{2125}{\MHz} DL~\cite{HuglSpatial}) and spatial information is congruent. We consider channels that can have frequency separation of several hundred percents, and hence some degree of spatial disagreement is expected.

To our knowledge, there is no prior work that uses the out-of-band information to aid the in-band mmWave covariance estimation.
Some other {out-of-band aided} mmWave communication methodologies, however, have been proposed. In~\cite{Raghavan2016Low}, coarse angle estimation at \subsGHz~followed by refinement at mmWave was proposed. In~\cite{Nitsche2015Steering}, the legacy WiFi measurements were used to configure the $60$~\GHz~WiFi links. The measurement results presented in~\cite{Nitsche2015Steering} demonstrated the benefits and practicality of using out-of-band information for mmWave communications. In~\cite{Ali2018Millimeter}, the \subsGHz~channel information was used to aid the beam-selection in analog mmWave systems. In~\cite{Hashemi2018Out}, a scheduling strategy for joint \subsGHz-mmWave communication system was introduced to maximize the delay-constrained throughput of the mmWave system. In~\cite{Gonzalez-Prelcic2016Radar}, radar aided mmWave communication was introduced. Specifically, the mmWave radar covariance was used directly to configure mmWave communication beams. 

The algorithms in~\cite{Raghavan2016Low,Ali2018Millimeter,Hashemi2018Out} were designed specifically for analog architectures. We consider a more general hybrid analog-digital architecture. Only LOS channels were considered in~\cite{Nitsche2015Steering}, whereas the methodologies proposed in this paper are applicable to NLOS channels. Radar information (coming from a band adjacent to the mmWave communication band) is used in~\cite{Gonzalez-Prelcic2016Radar}. We, however, use information from a \subsGHz~communication band as out-of-band information. 

In~\cite{Ali2016Estimating}, we provided preliminary results of covariance translation for SIMO narrowband systems and a single-cluster channel. In this paper, we extend~\cite{Ali2016Estimating} to consider multiple antennas at the transmitter (TX) and receiver (RX), multi-cluster frequency-selective wideband channels, and OFDM transmission, for both \subsGHz~and mmWave. Moreover, the preliminary results in~\cite{Ali2016Estimating} were obtained assuming a fully digital architecture at mmWave. We now consider a more practical hybrid analog-digital architecture. Finally,~\cite{Ali2016Estimating} did not include the out-of-band aided compressed covariance estimation strategy and the analytical bounds on the loss in received post-processing \SNR. 

The rest of the paper is organized as follows. In Section~\ref{sec:syschmodel}, we provide the system and channel models for \subsGHz~and mmWave. We present the out-of-band covariance translation in Section~\ref{sec:covtrans} and out-of-band aided compressed covariance estimation in Section~\ref{sec:comcovest}. In Section~\ref{sec:analysis}, we analyze the \SNR~degradation. We present the simulation results in Section~\ref{sec:simres}, and finally conclusions in Section~\ref{sec:conc}.

\textbf{Notation:} We use the following notation throughout the paper. Bold lowercase $\bx$ is used for column vectors, bold uppercase $\bX$ is used for matrices, non-bold letters $x$, $X$ are used for scalars. $[\bx]_i$, $[\bX]_{i,j}$, $[\bX]_{i,:}$, and $[\bX]_{:,j}$, denote $i$th entry of $\bx$, entry at the $i$th row and $j$th column of $\bX$, $i$th row of $\bX$, and $j$th column of $\bX$, respectively. We use serif font, e.g., $\bsfx$, for the frequency-domain variables. Superscript $\transp$, $\ast$ and $\dagger$ represent the transpose, conjugate transpose and pseudo inverse, respectively. $\bzero$ and $\bI$ denote the zero vector and identity matrix respectively. $\cC\cN(\bx,\bX)$ denotes a complex circularly symmetric Gaussian random vector with mean $\bx$ and covariance $\bX$. We use $\bbE[\cdot]$, $\|\!\cdot\!\|_p$, and $\|\!\cdot\!\|_\rmF$ to denote expectation, $p$ norm and Frobenius norm, respectively. The \subsGHz~variables are underlined, as $\ul{\bx}$, to distinguish them from mmWave.
\section{System, channel and, covariance models}\label{sec:syschmodel}
We consider a single-user multi-band MIMO system, shown in Fig.~\ref{fig:TXRX},where the \subsGHz~and mmWave systems operate simultaneously. We consider uniform linear arrays (ULAs) of isotropic point-sources at the TX and the RX. The strategies proposed in this work can be extended to other array geometries with suitable modifications. The \subsGHz~and mmWave arrays are co-located, aligned, and have comparable apertures.
\begin{figure}[h!]
\centering
\includegraphics[width=0.45\textwidth]{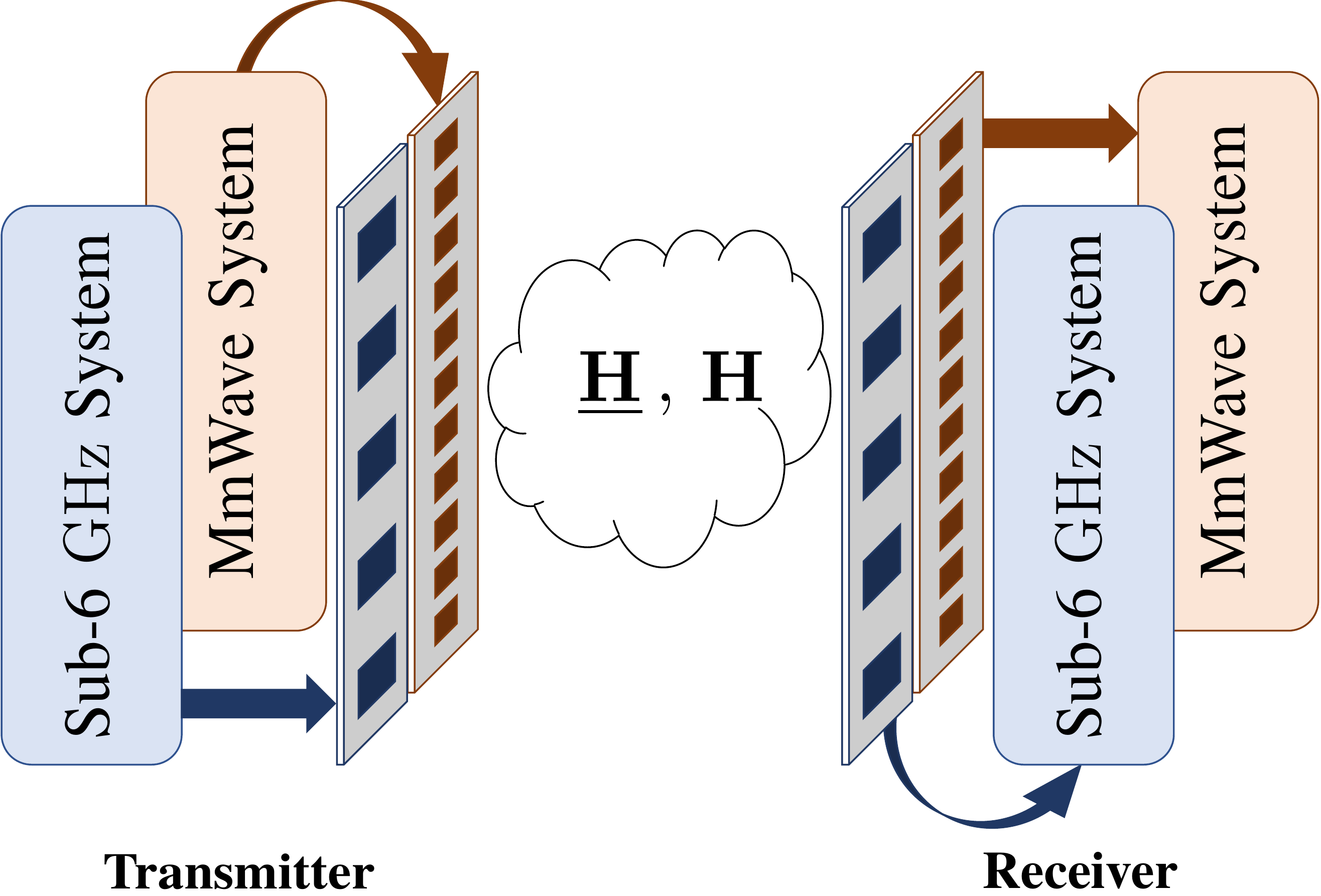}
\caption{The multi-band MIMO system with co-located \subsGHz~and mmWave antenna arrays. The \subsGHz~channel is denoted $\bHu$ and the mmWave channel is denoted $\bH$.}
\label{fig:TXRX}
\end{figure}
\subsection{Millimeter wave system model}
The mmWave system is shown in Fig.~\ref{fig:mmWavesystem}. The TX has $\NTX$ antennas and $\MTX\leq\NTX$ RF chains, whereas the RX has $\NRX$ antennas and $\MRX\leq\NRX$ RF chains. We assume that $\Ns\leq\min\{\MTX,\MRX\}$ data-streams are transmitted. We consider OFDM transmission with $K$ sub-carriers. The transmission symbols on sub-carrier $k$ are denoted as $\bsfs[k]\in\bbC^{\Ns\times1}$, and follow $\bbE[\bsfs[k]\bsfs^\ast[k]]=\frac{P}{K\Ns}\bI_{\Ns}$, where $P$ is the total average transmitted power. The data-symbols $\bsfs[k]$ are first precoded using the baseband-precoder $\bFBB[k]\in\bbC^{\MTX\times\Ns}$, then converted to time-domain using $\MTX$ $K$-point IDFTs. Cyclic-prefixes (CPs) are then prepended to the time-domain samples before applying the RF-precoder $\bFRF\in\bbC^{\NTX\times\MTX}$. Since the RF-precoder is implemented using analog phase-shifters, it has constant modulus entries i.e., $|[\bFRF]_{i,j}|^2=\frac{1}{\NTX}$. Further, we assume that the angles of the analog phase-shifters are quantized and have a finite set of possible values. With these assumptions, $[\bFRF]_{i,j}=\frac{1}{\sqrt{\NTX}}e^{\compj \zeta_{i,j}}$, where $\zeta_{i,j}$ is the quantized angle. The precoders satisfy the total power constraint $\sum_{k=1}^{K}\|\bFRF\bFBB[k]\|_\rmF^2=K\Ns$. 

We assume perfect time and frequency synchronization at the receiver. The received signals are first combined using the RF-combiner $\bWRF\in\bbC^{\NRX\times\MRX}$. The CPs are then removed and the time-domain samples are converted back to frequency-domain using $\MRX$ $K$-point DFTs. Subsequently, the frequency-domain signals are combined using the baseband combiner $\bWBB[k]\in\bbC^{\MRX\times\Ns}$. If $\bsfH[k]$ denotes the frequency-domain $\NRX\times\NTX$ mmWave MIMO channel on sub-carrier $k$, then the post-processing received signal on sub-carrier $K$ can be represented as 
\begin{align}
\bsfy[k]&=\bWBB^\ast[k]\bWRF^\ast\bsfH[k]\bFRF\bFBB[k]\bsfs[k]+\bWBB^\ast[k]\bWRF^\ast\bsfn[k],\nonumber\\
&=\bsfW^\ast[k]\bsfH[k]\bsfF[k]\bsfs[k]+\bsfW^\ast[k]\bsfn[k],
\label{eq:rxpost}
\end{align}
where $\bF[k]=\bFRF\bFBB[k]\in\bbC^{\NTX\times\Ns}$ is the precoder, and $\bW[k]=\bWRF\bWBB[k]\in\bbC^{\NRX\times\Ns}$ is the combiner. Finally, $\bsfn\sim\cC\cN(\bzero,\sigma_{\bsfn}^2\bI)$ is the additive white Gaussian noise. 

\begin{figure*}[h!]
\centering
\includegraphics[width=0.85\textwidth]{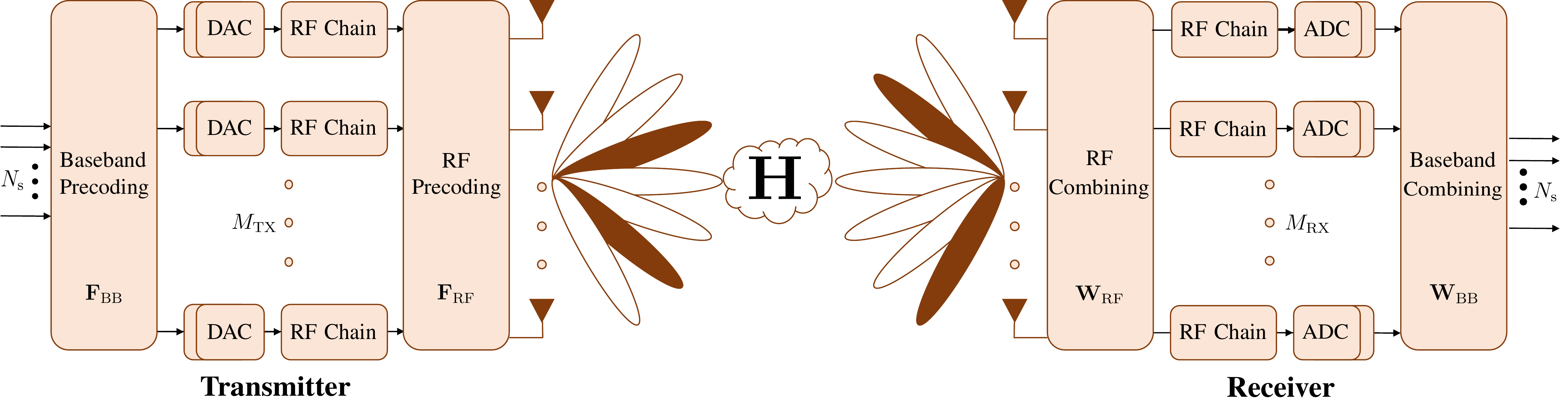}
\caption{The mmWave~system with hybrid analog/digital precoding.}
\label{fig:mmWavesystem}
\end{figure*}
\subsection{\SubsGHz~system model}
The \subsGHz~system is shown in Fig.~\ref{fig:sub6GHzsystem}. We underline all \subsGHz~variables to distinguish them from the mmWave variables. The \subsGHz~system has one RF chain per antenna and as such, fully digital precoding is possible. The $\ulNs$ data-streams are communicated by the TX with $\ulNTX$ antennas to the receiver with $\ulNRX$ antennas as shown in Fig.~\ref{fig:sub6GHzsystem}. The \subsGHz~OFDM system has $\ulK$ sub-carriers. 

\begin{figure*}[h!]
\centering
\includegraphics[width=0.6\textwidth]{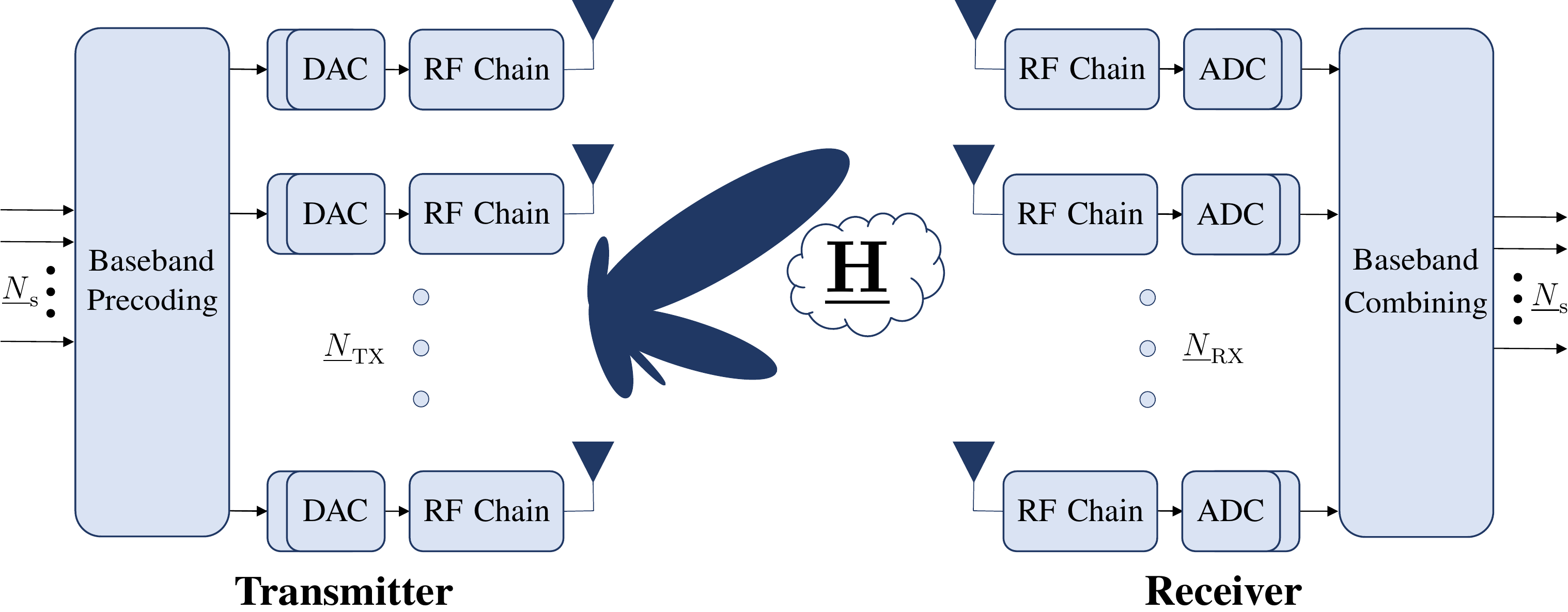}
\caption{The \subsGHz~system with digital precoding.}
\label{fig:sub6GHzsystem}
\end{figure*}
\subsection{Channel model}
We present the channel model for mmWave, i.e., using non-underlined notation. The \subsGHz~channel follows the same model. We adopt a wideband geometric channel model with $C$ clusters. Each cluster has a mean time-delay $\tau_c \in \bbR$, mean physical angle-of-arrival (AoA) and angle-of-departure (AoD) $\{\theta_c, \phi_c\} \in [0,2\pi)$. Each cluster is further assumed to contribute $R_c$ rays/paths between the TX and the RX. Each ray $r_c\in[R_c]$ has a relative time-delay $\tau_{r_c}$, relative angle shift $\{\vartheta_{r_c},\varphi_{r_c}\}$, and a complex path gain $\alpha_{r_c}$ (including path-loss). Further, $p(\tau)$ denotes the combined effects of analog filtering and pulse shaping filter evaluated at point $\tau$. Under this model, the delay-$d$ MIMO channel matrix $\bH[d]$ can be written as~\cite{Alkhateeb2016Frequency}
\begin{align}
\bH[d]=\sqrt{\NRX\NTX} &\sum_{c=1}^C \sum_{r_c=1}^{R_c} \alpharc p (d\Ts-\tau_c-\tau_{r_c})\times\nonumber\\
&\baRX(\theta_c+\vartheta_{r_c})\baTX^\ast(\phi_c+\varphi_{r_c}),
\label{eq:timedomch}
\end{align}
where $\Ts$ is the signaling interval and $\baRX(\theta)$ and $\baTX(\phi)$ are the antenna array response vectors of the RX and the TX, respectively. The array response vector of the RX is
\begin{align}
\baRX(\theta)=\frac{1}{\sqrt{\NRX}}[1,e^{\compj 2\pi\Delta\sin(\theta)},\cdots,e^{\compj(\NRX-1) 2\pi\Delta\sin(\theta)}]^\transp,
\end{align}
where $\Delta$ is the inter-element spacing normalized by the wavelength. The array response vector of the TX is defined in a similar manner. With the delay-$d$ MIMO channel matrix given in~\eqref{eq:timedomch}, the channel at sub-carrier $k$, $\bsfH[k]$ can be expressed as~\cite{Alkhateeb2016Frequency}
\begin{align}
\bsfH[k]=\sum_{d=0}^{D-1}\bH[d] e^{-\compj \tfrac{2\pi k}{K}d},
\label{eq:freqdomch}
\end{align}
where $D$ is the number of delay-taps in the mmWave channel.

\subsection{Covariance model}
We simplify~\eqref{eq:freqdomch} before discussing the channel covariance model. First, we plug in the definition of $\bH[d]$ from~\eqref{eq:timedomch} in~\eqref{eq:freqdomch}, change the order of summation, and re-arrange terms to write~\eqref{eq:freqdomch} as
\begin{align}
\bsfH[k]\!=\!\sqrt{\NRX\NTX} \sum_{c=1}^C \sum_{r_c=1}^{R_c}& \Big( \sum_{d=0}^{D-1}\alpharc p (d\Ts\!-\!\tau_c\!-\!\tau_{r_c})e^{-\compj \tfrac{2\pi k}{K}d} \Big)\nonumber\\
&\times\baRX^{}(\theta_c+\vartheta_{r_c})\baTX^\ast(\phi_c+\varphi_{r_c}).
\label{eq:fredomch2}
\end{align}
Second, we define $\bar{\alpha}_{r_{c},k}=\sum_{d=0}^{D-1}\alpharc p (d\Ts-\tau_c-\tau_{r_c})e^{-\compj \tfrac{2\pi k}{K}d}$ to rewrite~\eqref{eq:fredomch2}~as
\begin{align}
\bsfH[k]\!=\!\sqrt{\NRX\NTX} \sum_{c=1}^C \sum_{r_c=1}^{R_c} \bar{\alpha}_{r_{c},k}\baRX^{}(\theta_c\!+\!\vartheta_{r_c})\baTX^\ast(\phi_c\!+\!\varphi_{r_c}).
\label{eq:chsum} 
\end{align}
Finally, we define $\bar\balpha_k=[\bar\alpha_{1_1,k},\cdots,\bar\alpha_{R_1,k},\cdots,\bar\alpha_{1_C,k},\cdots,\bar\alpha_{R_C,k}]^\transp$, $\bARX=[\baRX(\theta_1+\vartheta_{1_1}),\cdots,$ $\baRX(\theta_C+\vartheta_{R_C})]$, and  $\bATX=[\baTX(\phi_1+\varphi_{1_1}),\cdots,\baTX(\phi_C+\varphi_{R_C})]$, to compactly write~\eqref{eq:chsum} as
\begin{align}
\bsfH[k]=\sqrt{\NRX\NTX} \bARX \diag(\bar\balpha_k) \bATX^\ast.
\label{eq:chdiag}
\end{align}

The transmit covariance of the channel on sub-carrier $k$ is defined as $\bsfRTX[k]=\frac{1}{\NRX}\bbE[\bsfH^\ast[k]\bsfH[k]]$ and the receive covariance is $\bsfRRX[k]=\frac{1}{\NTX}\bbE[\bsfH[k]\bsfH^\ast[k]]$. For the development of the proposed strategies, we make the typical assumption that for a given $k$, the gains $\bar{\alpha}_{r_{c},k}$ are uncorrelated and have variances $\sigma_{\bar\alpha_{r_c,k}}^2$. With this, the covariances across all sub-carriers are identical~\cite{Bjornson2009Exploiting}. In practice the channel delay-taps have some correlation and the covariances on all sub-carriers, though similar, are not identical. In Section~\ref{sec:simres}, we will test the robustness of the proposed strategies to the practical correlated delay-taps case. Under the uncorrelated gains simplification, the transmit covariance for fixed AoDs can be written as $\bsfRTX[k]=\NTX \bATX \bsfR_{\bar{\balpha}_k} \bATX^\ast$, and similarly $\bsfRRX[k]=\NRX\bARX^{} \bsfR_{\bar{\balpha}_k} \bARX^\ast$, where $\bsfR_{\bar{\balpha}_k}=\bbE[\bar\balpha_k^{}\bar\balpha_k^\ast]=\diag([\sigma^2_{\bar\alpha_{1_1,k}},\cdots,\sigma^2_{\bar\alpha_{R_C,k}}])$. 
 We denote the transmit covariance averaged across the sub-carriers simply as $\bsfRTX=\frac{1}{K}\sum_{k=1}^K \bsfRTX [k]$, and the averaged receive covariance as $\bsfRRX=\frac{1}{K}\sum_{k=1}^K \bsfRRX [k]$.
\section{Out-of-band covariance translation}\label{sec:intro}\label{sec:covtrans}
In this section, we address the problem of obtaining an estimate of the mmWave covariance directly from the \subsGHz~covariance with no in-band training. We continue the exposition assuming the receive covariance is translated (the transmit covariance is translated using the same procedure). To simplify notation, we remove the subscript $\rmR\rmX$ from the receive covariance in subsequent exposition. Hence, we seek to estimate $\bsfR\in\bbC^{\NRX\times\NRX}$ from $\ul{\bsfR}\in\bbC^{\ulNRX\times\ulNRX}$. We assume that the estimate of the \subsGHz~covariance $\hat{\ul \bsfR}$ is available. With no hardware constraints at \subsGHz~and a small number of antennas, empirical estimation of $\hat{\ul \bsfR}$ is easy~\cite{Czink2005Improved}. Further, the CSI at~\subsGHz~is required for the operation of the~\subsGHz~system itself. Therefore, obtaining the out-of-band information (i.e., the~\subsGHz~covariance) for mmWave covariance estimation does not incur any additional training overhead. 

In the parametric covariance translation proposed in this work, the parameters of the covariance matrix are estimated at \subsGHz. Subsequently, these parameters are used in the theoretical expressions of covariance matrices to generate mmWave covariance. To give a concrete example, consider a single-cluster channel. Assume a mean AoA $\theta$, AS $\sigma_\vartheta$, and a power azimuth spectrum (PAS) with characteristic function $\Phi(x)$ corresponding to $\sigma_\vartheta=1$. Then, under the small AS assumption, the channel covariance can be written as~\cite{Bengtsson2000Low}
\begin{align}
[\bsfR]_{i,j} = e^{\compj (i-j) 2 \pi \Delta \sin(\theta)} \Phi\big((i-j) 2 \pi \Delta  \cos(\theta) \sigma_{\vartheta} \big).
\label{eq:RXcovariance}
\end{align}   

To get a closed form expression for the covariance, \eqref{eq:RXcovariance}~is evaluated for a specific PAS. The resulting expressions for truncated Laplacian, truncated Gaussian, and Uniform distributions are summarized in Table~\ref{tab:ThExp}. For a single-cluster channel, the mean AoA and AS of the cluster (i.e., only two parameters) are estimated at \subsGHz~and subsequently used in one of the expressions (in Table~\ref{tab:ThExp}) to obtain the mmWave covariance~\cite{Ali2016Estimating}. For channels with multiple-clusters, the parametric covariance translation is complicated as the number of unknown parameters is typically higher. As an example, for only a two-cluster channel, $6$ parameters need to be estimated. The $6$ parameters are the AoA and AS of both clusters (i.e., $4$ parameters), and the power contribution of each cluster in the covariance (i.e., $2$ additional parameters). The estimation procedure is further complicated by the fact that the number of clusters is unknown, and needs to be estimated. In the following, we outline a parametric covariance translation procedure for multi-cluster channels.

\begin{table*}[h!]\centering
\caption{Theoretical expressions for covariance $[\bsfR]_{i,j}$}
\label{tab:ThExp}
\begin{tabular}{|c|c|}
\hline
\textbf{PAS} & \textbf{Expression}\\
\hline
\rule{0pt}{4ex} Truncated Laplacian~\cite{Forenza2007Simplified} & $\frac{\beta e^{\compj 2\pi \Delta (i-j) \sin(\theta)}}{1+\frac{\sigma_\vartheta^2}{2}[2\pi \Delta (i-j) \cos(\theta)]^2},~\beta=\frac{1}{1-e^{-\sqrt{2}\pi/\sigma_\vartheta}}$\\
\hline
\rule{0pt}{4ex} Truncated Gaussian~\cite{Bengtsson2000Low} & $ e^{-((i-j) 2 \pi \Delta  \cos(\theta) \sigma_{\vartheta})^2} e^{\compj 2\pi \Delta (i-j) \sin(\theta)}$ \\
\hline
\rule{0pt}{4ex} Uniform~\cite{Bengtsson2000Low} & $ \frac{\sin((i-j)\varrho_\vartheta)}{((i-j)\varrho_\vartheta)} e^{\compj  2 \pi \Delta (i-j) \sin(\theta)},~\varrho_\vartheta=\sqrt{3}\times2\pi\Delta\sigma_\vartheta\cos(\theta)$\\
\hline
\end{tabular}
\end{table*}

For clarity in exposition, we consider the covariance translation to be a four-step procedure and explain each step separately. In the first three steps, the parameters are estimated from \subsGHz~covariance. These parameters are: (i) the number of clusters, (ii) the AoA and AS of each cluster, and (iii) the power contribution of each cluster in the covariance. The fourth step uses the estimated parameters to obtain the mmWave covariance. 
\subsubsection{Estimating the number of clusters }
The first step in the parametric translation is to estimate the number of clusters in the channel. Enumerating the number of signals impinging on an array is a fundamental problem known as model order selection. The most common solution is to use information theoretic criteria e.g., minimum description length (MDL)~\cite{Wax1985Detection} or Akaike information criterion (AIC)~\cite{Akaike1974new}. The model order selection algorithms estimate the number of point-sources and do not directly give the number of clusters (i.e., distributed/scattered sources). To obtain the number of clusters, we make the following observation. The dimension of the signal subspace of a covariance matrix corresponding to a two point-source channel is $2$. In addition, it was shown in~\cite{Bengtsson2000Low} that, the dimension of the signal subspace of the covariance matrix corresponding to a channel with a single-cluster and small AS is also $2$. With this observation, the model order selection algorithms can be used for estimating the number of clusters. Specifically, if the number of point-sources estimated by a model order selection algorithm is $\hat{\ul{\rmP\rmS}}$, we consider the channel to have $\hat{\ul{C}}=\max\{\lfloor\frac{\hat{\ul{\rmP\rmS}}}{2}\rfloor,1\}$ clusters. The term $\lfloor\frac{\hat{\ul{\rmP\rmS}}}{2}\rfloor$ equates the number of clusters to half the point-sources (exactly for even number of point-sources, and approximately for odd). We set the minimum number of clusters to $1$ to deal with the case of a single source with very small AS. 

\subsubsection{Estimating angle-of-arrival and angle spread}
Prior work has considered the specific problem of estimating both the AoA and the AS jointly from an empirically estimated spatial covariance matrix (e.g., maximum likelihood estimation~\cite{Trump1996Estimation}, covariance matching estimation~\cite{Ottersten1998Covariance}, and spread root-MUSIC estimation~\cite{Bengtsson2000Low}). We use spread root-MUSIC algorithm due to its low computational complexity and straightforward extension for multiple-clusters. We refer the interested reader to~\cite{Bengtsson2000Low} for the details of the spread root-MUSIC algorithm. Here, we focus instead on a robustification necessary for the success of the proposed strategy.

If the channel has a single-cluster and the AS is very small, the spread root-MUSIC algorithm can fail~\cite{Bengtsson2000Low}. In this case, the algorithm returns an arbitrary AoA and an unusually large AS. This failure can be detected by setting a threshold on AS. Specifically, if the estimated AS is larger than the threshold value, AoA only estimation (e.g., using root-MUSIC~\cite{Barabell1983Improving}) is performed and the AS is set to zero. In addition, the AoA only estimation should also be performed when only a single point-source is detected while estimating the number of clusters.

\subsubsection{Estimating the power contribution of each cluster}
We denote the covariance due to the $\ul{c}$th cluster as $\ul \bsfR(\ul{\theta}_{\ul{c}},\ul{\sigma}_{\ul{\vartheta},\ul{c}})$. This covariance is calculated using the expressions in~Table~\ref{tab:ThExp}. Specifically, the AoA and AS estimated from the second step are used, and the covariance expressions are evaluated for the number of antennas in the \subsGHz~system. Further, we denote the power contribution of the $\ul c$th cluster as $\ul\epsilon_{\ul c}$. Now, under the assumption of uncorrelated clusters, the total covariance can be written as 
\begin{align}
\ul \bsfR=\sum_{\ul{c}=1}^{\ul C}{\ul \epsilon}_{\ul c}\ul \bsfR(\ul \theta_{\ul c},\ul\sigma_{\ul \vartheta,\ul c}) + \ul\sigma_{\ul \bsfn}^2\bI,
\label{eq:totalcovmatrix}
\end{align}
Introducing the vectorized notation $\bsfr=\vec(\bsfR)$ for the covariance matrix, we re-write \eqref{eq:totalcovmatrix}~as
\begin{align}
\ul \bsfr=\left[\ul \bsfr(\ul \theta_1,\ul \sigma_{\ul \vartheta,1}),\cdots,\ul \bsfr(\ul \theta_{\ul C},\ul \sigma_{\ul \vartheta,\ul C}),  \vec(\bI)\right] 
\begin{bmatrix}
\ul\epsilon_1\\
\vdots\\
\ul\epsilon_{\ul C}\\
\ul\sigma_{\ul \bsfn}^2
\end{bmatrix}.
\label{eq:totalcovvector}
\end{align}
The system of equations~\eqref{eq:totalcovvector} can be solved (e.g., using non-negative least-squares) to obtain the power
contributions of the clusters.
\subsubsection{Obtaining the mmWave covariance}
The mmWave covariance corresponding to the $\ul{c}$th cluster is denoted as $\bsfR(\ul{\theta}_{\ul{c}},\ul{\sigma}_{\ul{\vartheta},\ul{c}})$. Similar to \subsGHz~covariance $\ul\bsfR(\ul{\theta}_{\ul{c}},\ul{\sigma}_{\ul{\vartheta},\ul{c}})$, the mmWave covariance $\bsfR(\ul{\theta}_{\ul{c}},\ul{\sigma}_{\ul{\vartheta},\ul{c}})$ is also calculated using the expressions in~Table~\ref{tab:ThExp}. The covariance expressions, however, are now evaluated for the number of antennas in the mmWave system. With this, we have the mmWave covariances corresponding to all ${\ul C}$ clusters. Further, we have estimates of the cluster power contributions ${\ul\epsilon}_{\ul c}$ from step three. We now use the mmWave analog of~\eqref{eq:totalcovmatrix} to obtain the mmWave covariance, i.e.,
\begin{align}
\bsfR=\sum_{\ul c=1}^{\ul C}{\ul\epsilon}_{\ul c}\bsfR(\ul\theta_{\ul c},\ul\sigma_{\ul \vartheta,\ul c}).
\label{eq:totalcovmatrixmmWave}
\end{align}

We have purposefully ignored the contribution of white noise in~\eqref{eq:totalcovmatrixmmWave}. Though it is possible to estimate the noise variance at mmWave, it is not necessary for our application. This is because the hybrid precoders/combiners are designed to approximate the dominant singular vectors of the channel covariance matrix~\cite{Alkhateeb2013Hybrid}. As the singular vectors of a covariance matrix do not change with the addition of a scaled identity matrix, the addition is inconsequential.
\section{Out-of-band aided compressed covariance estimation}\label{sec:comcovest}
In this section, we formulate the problem of compressed covariance estimation in hybrid mmWave systems. There is some prior work on covariance estimation in hybrid mmWave systems, see e.g.,~\cite{Mendez-Rial2015Adaptive,Park2016Spatial}. In~\cite{Park2016Spatial}, the Hermitian symmetry of the covariance matrix and the limited scattering of the mmWave channel are exploited. By exploiting Hermitian symmetry, \cite{Park2016Spatial} outperforms the methods that only use sparsity e.g., \cite{Mendez-Rial2015Adaptive}. We closely follow the framework of~\cite{Park2016Spatial} for compressed covariance estimation. As only SIMO systems were considered in~\cite{Park2016Spatial}, we extend~\cite{Park2016Spatial} to MIMO systems. Subsequently use the concepts of weighted sparse signal recovery to aid the in-band compressed covariance estimation with out-of-band information. 
\subsection{Problem formulation}
We start with an implicit understanding that the formulation is per sub-carrier, but do not explicitly mention $k$ in the equations to reduce the notation overhead. We assume a single stream transmission in the training phase without loss of generality.  With $\Ns=1$, the post RF-combining received signal can be written as
\begin{align}
\bsfy_t=\bWRFt^\ast\bsfH_t\bff + \bWRFt^\ast \bsfn_t,
\label{eq:rxdcombfirst}
\end{align}
where we have introduced a discrete time index $t$. The time index $t$ denotes a snapshot. We assume that the channel remains fixed inside a snapshot. Further, we have used vector notation for the precoder to highlight the single stream case and have a made a simplistic choice $\sfs_t=1$ for ease of exposition. 

For our application, a single snapshot consists of two consecutive OFDM training frames. This is because the transmitter can synthesize an omni-directional precoder using two consecutive transmissions. An example is that in the first training frame, we use $\bff_1=\frac{1}{\sqrt{\NTX}}[1,\cdots,1]^\transp$ , and in the second we use $\bff_2=\frac{1}{\sqrt{\NTX}}[1,-1,\cdots,-1]^\transp$. To see how these precoders can give omni-directional transmission, we write the received signal in the first transmission of the $t$th snapshot as
\begin{align}
\bsfy_{t,1}=\bWRFt^\ast\bsfH_t\bff_1 + \bWRFt^\ast \bsfn_{t,1},
\label{eq:rxdcombfirst1}
\end{align}
where $\bsfn_{t,1}\sim\cC\cN(\bzero,\sigma_{\bsfn}^2\bI)$, and the received signal in the second transmission of the $t$th snapshot as 
\begin{align}
\bsfy_{t,2}=\bWRFt^\ast\bsfH_t\bff_2 + \bWRFt^\ast \bsfn_{t,2}.
\label{eq:rxdcombfirst2}
\end{align}
Now we consider the received signal~\eqref{eq:rxdcombfirst} in the $t$th snapshot, as the sum of the two individual transmissions, i.e.,
\begin{align}
\bsfy_t&=\bsfy_{t,1}+\bsfy_{t,2}=\bWRFt^\ast\bsfH_t(\bff_1+\bff_2) + \bWRFt^\ast(\bsfn_{t,1}+\bsfn_{t,2}),\nonumber\\
&=\frac{2}{\sqrt{\NTX}}\bWRFt^\ast\bsfH_t[1,0,\cdots,0]^\transp + \bWRFt^\ast(\bsfn_{t,1}+\bsfn_{t,2}).
\label{eq:rxdcombfirst3}
\end{align}
Thus effectively, combined over two transmissions, the precoder behaves as an omni-directional precoder, and effectively reduces a MIMO system to a SIMO system. The factor $\dfrac{2}{\sqrt{\NTX}}$ in~\eqref{eq:rxcombgt} denotes the power lost in trying to achieve omni-directional transmission. Similarly, as two independent transmissions are summed up, we have $\bsfn_t\sim\cC\cN(\bzero,2\sigma_{\bsfn}^2\bI)$. Depending on the scenario, this lost in \SNR~(due to low received power and increased noise variance) may be tolerated or compensated by repeated transmission. Assuming that the path angles do not change during the $T$ snapshots, the MIMO channel~\eqref{eq:chdiag} can be written as 
\begin{align}
\bsfH_t=\sqrt{\NRX\NTX}\bARX^{}\diag(\bar\balpha_t)\bATX^\ast,~t=1,2,\cdots,T.
\end{align}
and further
\begin{align}
\bsfH_t[1,0,\cdots,0]^\transp&=\sqrt{\NRX}\bARX^{}\bar\balpha_t.
\end{align}
Now, the received signal~\eqref{eq:rxdcombfirst} can be simply re-written as
\begin{align}
\bsfy_t=\bWRFt^\ast\bARX \bsfg_t + \bWRFt^\ast \bsfn_t.
\label{eq:rxcombgt}
\end{align}
where we have introduced $\bsfg_t=2\sqrt{\dfrac{\NRX}{\NTX}}\bar\balpha_t$. After which, the covariance of the received signal $\bsfy_t$ is
\begin{align}
\bsfRy=\bbE[\bsfy\bsfy^\ast]=\bWRF^\ast\bARX\bsfRg\bARX^\ast\bWRF+2\sigma_{\bsfn}^2\bWRF^\ast\bWRF,
\end{align}
where $\bsfRg=\bbE[\bsfg^{}\bsfg^\ast]$. By the definition of $\bsfg_t$, we have
\begin{align}
\bsfRg&=4\frac{\NRX}{\NTX}  \bbE [ \bar\balpha \bar\balpha^\ast ]=4\frac{\NRX}{\NTX} \bR_{\bar\balpha}.
\label{eq:Rg}
\end{align}

As the RX covariance can be written as $\bsfRRX=\NRX\bARX \bsfR_{\bar\balpha} \bARX^\ast$, once $\bsfRg$ and the AoAs are estimated, the receive covariance can be obtained. Hence, the main problem is to recover $\bsfRg$ and the AoAs from $\bsfRy$. To do so, we re-write~\eqref{eq:rxcombgt} as
\begin{align}
\bsfy_t\approx\bWRFt^\ast\bbARX \bar\bsfg_t + \bWRFt^\ast \bsfn_t.
\label{eq:rxcombgtapprox}
\end{align}
where $\bbARX$ is a $\NRX \times B_{\rmR\rmX}$ dictionary matrix whose columns are composed of the array response vector associated with a predefined set of AoAs, and $\bar\bsfg_t$ is a $B_{\rmR\rmX}\times 1$ vector. The approximation in~\eqref{eq:rxcombgtapprox} appears as the true AoAs in the channel are not confined to the predefined set. Note that even though there are several paths in the channel, due to clustered behavior, the AoAs are spaced closely and hence the number of coefficients with significant magnitude in $\bar\bsfg_t$ is $L\ll B_{\rmR\rmX}$. 

Due to limited scattering of the channel, the matrix, $\bar\bsfg_t\bar\bsfg_t^\ast$, has a Hermitian sparse structure. This structure can be exploited in the estimation of $\hat{\bsfR}_{\bar \bsfg}$ via the algorithm called covariance OMP (COMP)~\cite{Park2016Spatial}. The performance of the COMP algorithm, however, is limited by the number of RF chains used in the systems. This limitation can be somewhat circumvented by using time-varying RF-combiners $\bWRFt$~\cite{Alkhateeb2014Channel,Park2016Spatial}. Specifically, we use a distinct RF-combiner in each snapshot. 
The modification of COMP that uses time-varying RF-combiners is called dynamic covariance OMP (DCOMP)~\cite{Park2016Spatial}. 

\textbf{Remark:} Our extension of~\cite{Park2016Spatial} (from SIMO to MIMO systems) is based on omni-directional precoding to reduce the MIMO system to a SIMO system. Another possible extension of~\cite{Park2016Spatial} to MIMO systems was outlined in~\cite{Park2017Spatial}. Specifically, the full MIMO covariance~$\bR_{\rm full}=\bbE[\vec(\bH)\vec(\bH)^\ast]$ was estimated in~\cite{Park2017Spatial}, though with high computational complexity. To understand this, consider $\NTX=\NRX=64$ antennas and  $4$x oversampled dictionaries i.e., $B_{\rmR\rmX}=B_{\rmT\rmX}=256$. These are modest system parameters for mmWave communication and were used in~\cite{Park2017Spatial}. With these parameters, the full covariance estimation requires support search over a  $B_{\rmR\rmX}B_{\rmT\rmX}\times B_{\rmR\rmX}B_{\rmT\rmX}=65536\times 65536$ dimensional Hermitian-sparse unknown. In comparison, our approach requires the recovery of a $B_{\rmR\rmX}\times B_{\rmR\rmX}$ unknown and a $B_{\rmT\rmX}\times B_{\rmT\rmX}$ unknown, i.e., two $256\times 256$ dimensional Hermitian-sparse unknowns. Furthermore, in mmWave systems, the precoders and combiners used at mmWave are designed based on transmit and receive covariances separately. Therefore, our approach is more appropriate than~\cite{Park2017Spatial}.  
\subsection{Weighted compressed covariance estimation}
The compressed covariance estimation algorithm divides the AoA range into $B_{\rmR\rmX}$ intervals using the dictionary $\bbARX$ and assumes that the prior probability of the support is uniform, i.e., the active path angles on the grid have the same probability $p$ throughout the AoA range. This is a reasonable assumption under no prior information about the AoAs. If some prior information about the non-uniformity in the support is available, the compressed covariance estimation algorithms can be modified to incorporate this prior information. Note that the DCOMP algorithm is an extension of the OMP algorithm to the covariance estimation problem. In~\cite{Scarlett2013Compressed} a modified OMP algorithm called logit weighted - OMP (LW-OMP) was proposed for non-uniform prior probabilities. Here we use logit weighting in compressed covariance estimation via DCOMP algorithm. Assume that $\brho \in \bbR^{B_{\rmR\rmX}\times 1}$ is the vector of prior probabilities $0 \leq [\brho]_i \leq 1$. Then we introduce an additive weighting function $w([\brho]_i)$ to weight the DCOMP algorithm according to prior probabilities. The authors refer the interested reader to~\cite{Scarlett2013Compressed} for the details of logit weighting and the selection of $w([\brho]_i)$. The general form of $w([\brho]_i)$, however, can be given as $w([\brho]_i)=J_{\rmw} \log\dfrac{[\brho]_i}{1-[\brho]_i}$, where $J_{\rmw}$ is a constant that depends on the number of active coefficients in $\hat{\bsfR}_{\bar \bsfg}$, the amplitude of the unknown coefficients, and the noise level~\cite{Scarlett2013Compressed}. We present the logit weighted - DCOMP (LW-COMP) in Algorithm \ref{alg:COMP}. In the absence of prior information, LW-DCOMP can be solved using uniform probability $\brho=\varepsilon\bone$, where $0<\varepsilon<=1$, which is equivalent to DCOMP.

\begin{algorithm}
\caption{Logit weighted - Dynamic Covariance OMP (LW-DCOMP)}
\begin{algorithmic}[1]
\INPUT $\bWRFt\forall T,\bsfy_t\forall T, \bbARX,\sigma_{\bsfn}^2,\bp$
\INIT $\bV_t=\bsfy_t\bsfy_t^\ast\forall t, \cS=\emptyset, i=0$
\WHILE{($\sum_t\|\bV_t\|_\rmF > 2\sigma_{\bsfn}^2\sum_t\|\bWRFt^\ast\bWRFt\|_\rmF$ and $i<\MRX$)}
\STATE $j=\arg\max_i \sum_{t=1}^{T} | [\bWRFt\bbARX]^\ast_{:,i}\bV_t [\bWRFt\bbARX]^{}_{:,i} | + w([\brho]_i)$
\STATE $\cS=\cS\cup\{j\}$
\STATE $\hbRbgt=[\bWRFt\bbARX]^\dagger_{:,\cS}(\bsfy_t\bsfy_t^\ast)\big([\bWRFt\bbARX]^\dagger_{:,\cS}\big)^\ast,~\forall t$
\STATE $\bV_t=\hbRyt-[\bWRFt\bbARX]^{}_{:,\cS}\hbRbgt[\bWRFt\bbARX]^\ast_{:,\cS},\forall t$
\STATE $i=i+1$
\ENDWHILE
\OUTPUT $\cS$, $\hbRbg=\frac{1}{T}\sum_{t=1}^{T}\hbRbgt$.
\end{algorithmic}
\label{alg:COMP}
\end{algorithm}

The spatial information from \subsGHz~can be used to obtain a proxy for $\brho$. Specifically, let us define an $\ulNRX\times B_{\rmR\rmX}$ dictionary matrix $\ul{\bar{\bA}}_{\rmR\rmX}$, which is obtained by evaluating the \subsGHz~array response vectors at the same points as the mmWave array response vector is evaluated to get the dictionary matrix $\bbARX$. Then a simple proxy of the probability vector based on the~\subsGHz~covariance can be obtained as follows
\begin{align}
\brho=J_{\uprho} \frac{|\frac{1}{B_{\rmR\rmX}} \sum_{b=1}^{B_{\rmR\rmX}} [\ul{\bar{\bA}}_{\rmR\rmX}^\ast~\ul\bsfR~\ul{\bar{\bA}}_{\rmR\rmX}]_{:,b}|}{\max |\frac{1}{B_{\rmR\rmX}} \sum_{b=1}^{B_{\rmR\rmX}} [\ul{\bar{\bA}}_{\rmR\rmX}^\ast~\ul\bsfR~\ul{\bar{\bA}}_{\rmR\rmX}]_{:,b}|},
\label{eq:rhoproxy}
\end{align}
where $J_{\uprho}$ is an appropriately chosen constant that captures the reliability of the out-of-band information. The reliability is a function of the \subsGHz~and  mmWave spatial congruence, and operating \SNR. A higher value for $J_{\uprho}$ should be used for for highly reliable information. For the results in Section~\ref{sec:simres}, we optimized for $J_{\uprho}$ by testing a few values and choosing the one that gave the best performance.
\section{\SNR~degradation due to covariance mismatch}\label{sec:analysis}
We start by providing the preliminaries required for analyzing the loss in received post-processing \SNR~due to imperfections in channel covariance estimates. We perform the analysis for a single path channel and as such single stream transmission suffices. This can be considered an extreme case where the AS is zero, and as such the only AoA is the mean AoA $\theta$. For the channel model presented in Section~\ref{sec:syschmodel}, this implies that the receive covariance can be written as $\bsfRRX=\NRX\sigma_\alpha^2\baRX(\theta)\baRX^\ast(\theta)$. Similarly, the transmit covariance can be written as $\bsfRTX=\NTX\sigma_\alpha^2\baTX(\phi)\baTX^\ast(\phi)$. The subspace decomposition of the receive covariance matrix is
\begin{align}
\bsfRRX=\bURX\bSigma\bURX^\ast=\bURXs\bSigmas\bURXs^\ast + \bURXn\bSigman\bURXn^\ast,
\label{eq:subsapcetrue}
\end{align}
where the columns of $\bU_{\rmR\rmX}$ are the singular vectors and the diagonal entries of $\bSigma$ are the singular values. Further, the columns of $\bURXs$ are the singular vectors associated with non-zero singular values and span the signal subspace. For a single path channel, there is only one non-zero singular value and only one singular vector spanning the signal subspace. The columns of $\bURXn$ are the singular vectors associated with zero singular values and span the noise subspace. The transmit covariance $\bsfRTX$ also has a similar subspace decomposition. For statistical digital precoding (combining), the singular vectors of the transmit (receiver) covariance matrices are used as precoders (combiners). For a single path channel, the array response vector evaluated at AoA $\theta$ is a valid singular vector of $\bsfRRX$. As such, the received signal with digital precoding/combing based on singular vectors of the covariance matrices can be written as
\begin{align}
\sfy&=\bURXs^\ast \bsfH \bUTXs \sfs + \bURXs^\ast\bsfn,\nonumber\\
&=\sqrt{\NRX\NTX}\alpha \baRX^\ast(\theta) \baRX(\theta) \baTX^\ast(\phi)  \baTX(\phi) \sfs + \bURXs^\ast \bsfn,\nonumber\\
&=\sqrt{\NRX\NTX}\alpha \sfs + \bURXs^\ast\bsfn.
\label{eq:rxdtruecov}
\end{align}

From~\eqref{eq:rxdtruecov}, the average received \SNR~with perfect covariance knowledge is
\begin{align}
\SNR_{\bsfR}=\dfrac{\NRX\NTX \bbE[|\alpha|^2]\bbE[|\sfs|^2]}{\bbE[\|\bURXs^\ast\bsfn\|_\rmF^2]}.
\label{eq:SNRwithperfectcov}
\end{align}
Recall from Section~\ref{sec:syschmodel} that the variance of channel paths is $\bbE[|\alpha|^2]= \sigma_\alpha^2$, and the transmit symbol power is $\bbE[|\sfs|^2]=\frac{P}{K}$.  Further, with noise $\bsfn\sim\cC\cN(\bzero,\sigma_{\bsfn}^2\bI)$, we have $\bbE[\|\bURXs^\ast\bsfn\|_\rmF^2]=\sigma_{\bsfn}^2$. Therefore, we re-write~\eqref{eq:SNRwithperfectcov} as
\begin{align}
\SNR_{\bsfR}=\dfrac{\NRX\NTX P \sigma_\alpha^2}{K \sigma_{\bsfn}^2}.
\label{eq:SNRwithperfectcovsimplified}
\end{align}

We model the error in the estimated covariance as additive, i.e., the true covariance matrix and the estimated covariance matrix differ by a perturbation $\delbsfR$ such that $\hat{\bsfR}_{\rmR\rmX}=\bsfRRX+\delbsfRRX$ and $\hat{\bsfR}_{\rmT\rmX}=\bsfRTX+\delbsfRTX$. A decomposition of the estimated covariance matrix $\hat{\bsfR}_{\rmR\rmX}$  (similar to~\eqref{eq:subsapcetrue}) is 
\begin{align}
\hat{\bsfR}_{\rmR\rmX}=\hbURX \hbSigma \hbURX^\ast=\hbURXs\hbSigmas\hbURXs^\ast + \hbURXn\hbSigman\hbURXn^\ast.
\label{eq:subspaceest}
\end{align}
The vector $\hbURXs^\ast$ can be written as a sum of two vectors $\bURXs$ and $\delURXs$. Hence, the vector $\hbURXs$ will typically not meet the normalization $\|\hbURXs\|^2=1$ assumed in the system model. We ensure that the power constraint on the precoders/combiners is met by using a normalized version. Hence, the received signal with digital precoding/combining based on the imperfect covariance is
\begin{align}
\sfy= \frac{\hbURXs^\ast}{\| \hbURXs \|} \bsfH \frac{\hbUTXs^\ast}{\| \hbUTXs \|} \sfs + \frac{\hbURXs^\ast}{\| \hbURXs \|}\bsfn.
\label{eq:rxdestcov}
\end{align}
Now we quantify the averaged receive \SNR~with imperfect covariance knowledge in the following theorem.

\begin{thm}
For the received signal $\sfy$ in~\eqref{eq:rxdestcov}, the precoder that follows the model $\hbUTXs=\hbUTXs+\delUTXs$, and the combiner that follows the model $\hbURXs=\hbURXs+\delURXs$, the averaged received SNR~is
\begin{align}
\SNR_{\hat{\bsfR}}=\dfrac{\NRX\NTX P \sigma_\alpha^2}{K \sigma_{\bsfn}^2 \| \hbURXs \|^2 \| \hbUTXs \|^2}.
\label{eq:SNRest}
\end{align}
\end{thm}
\begin{proof}
See Appendix A.
\end{proof}

Now given the \SNR~with perfect covariance~\eqref{eq:SNRwithperfectcovsimplified} and the \SNR~with imperfect covariance ~\eqref{eq:SNRest}, the loss in the \SNR, $\gamma$ is
\begin{align}
\gamma=\frac{\SNR_{\bsfR}}{\SNR_{\hat{\bsfR}}}=\| \hbURXs \|^2 \| \hbUTXs \|^2.
\label{eq:snrlosssigsub}
\end{align}

The \SNR~loss~\eqref{eq:snrlosssigsub} is given in terms of the vectors that span the signal subspace of the estimated covariance matrices. In the following theorem, we give the loss explicitly in terms of the perturbations $\delbsfRRX$ and $\delbsfRTX$.

\begin{thm}
The loss in the SNR~$\gamma$ can be written approximately as
\begin{align}
\gamma\approx\left(1+\frac{\|\delbsfRRX\bURXs\|^2}{\NRX^2\sigma_\alpha^4}\right)\left(1+\frac{\|\delbsfRTX\bUTXs\|^2}{\NTX^2\sigma_\alpha^4}\right),
\label{eq:snrlossapprox}
\end{align}
and can be bounded as
\begin{align}
\gamma\lesssim\left(1+\frac{\sigma_{\max}^2(\delbsfRRX)}{\NRX^2\sigma_\alpha^4}\right)\left(1+\frac{\sigma_{\max}^2(\delbsfRTX)}{\NTX^2\sigma_\alpha^4}\right),
\label{eq:snrlossupper}
\end{align}
and
\begin{align}
\gamma\gtrsim\left(1+\frac{\sigma_{\min}^2(\delbsfRRX)}{\NRX^2\sigma_\alpha^4}\right)\left(1+\frac{\sigma_{\min}^2(\delbsfRTX)}{\NTX^2\sigma_\alpha^4}\right),
\label{eq:snrlosslower}
\end{align}
where $\sigma_{\max}(\cdot)$ and $\sigma_{\min}(\cdot)$ represent the largest and smallest singular value of the argument.
\end{thm}
\begin{proof}
See Appendix B.
\end{proof}

The \SNR~loss expression~\eqref{eq:snrlossapprox} admits a simple explanation. The loss is proportional to the alignment of the true signal subspace to the column space of the perturbation matrix. If the true signal subspace is orthogonal to the column space of the perturbation matrix i.e., $\bURXs$ lies in the null space of $\delbsfRRX$, then there is no loss in the \SNR~due to the perturbation, which makes intuitive sense. Further, the results~\eqref{eq:snrlossupper} and~\eqref{eq:snrlosslower} give the bounds on \SNR~loss explicitly in the form of the singular values of the perturbation matrices $\delbsfRRX$ and $\delbsfRTX$.
\section{Simulation results}\label{sec:simres}
In this section, we present simulation results to test the effectiveness of the proposed covariance estimation strategies and validate the \SNR~loss analysis. First, we test the performance of the proposed covariance estimation strategies in simpler channels, assuming that the parameters governing the \subsGHz~and mmWave channels are consistent. This is to say that the cluster in the \subsGHz~and mmWave channel has the same AoA, AoD, arrival AS, and departure AS. Subsequently, we study the performance of the proposed covariance estimation strategies in realistic channels when the parameters of the \subsGHz~and mmWave channels do not match. Finally, we validate the \SNR~loss analysis. To show the benefit of the out-of-band information in comparison with in-band only training, we compare the proposed strategies with the DCOMP algorithm~\cite{Park2016Spatial}. For covariance estimation, the DCOMP algorithm was shown to perform better than several well known sparse recovery algorithms~\cite{Park2016Spatial}.

We test the performance of the proposed covariance estimation strategies using two metrics. The first metric is the efficiency metric $\eta$~\cite{Haghighatshoar2017Massive} that captures the similarity in the signal subspace of the true covariance and the estimated covariance. This metric is relevant in the current setup as the precoder/combiners are designed using the singular vectors that span the signal subspace. The efficiency metric is given as~\cite{Haghighatshoar2017Massive}
\begin{align}
\eta(\bsfR,\hat{\bsfR})=\frac{\tr(\bU_{\Ns}^\ast\hat{\bsfR}\bU_{\Ns}^{})}{\tr(\hat\bU_{\Ns}^\ast\bsfR\hat\bU_{\Ns})},
\label{eq:effmtc}
\end{align}
where $\bU_{\Ns}~(\hat \bU_{\Ns})$ are the $\Ns$ singular vectors of the matrix $\bsfR~(\hat{\bsfR})$ corresponding to the largest $\Ns$ singular values. It is clear that $0\leq\eta\leq1$ and it is desirable to make $\eta$ as close to $1$ as possible. 

The second metric is the achievable rate using the hybrid precoders/combiners designed from the covariance information. We assume that the channel covariance is constant over $T_{\rm{stat.}}$ OFDM blocks. Here, the subscript ${\rm{stat.}}$ signifies that the interval $T_{\rm{stat.}}$ is the time for which the statistics remain unchanged, and not the coherence time of the channel. In fact, the statistics vary slowly and typically $T_{\rm{stat.}}$ is much larger than the channel coherence time. If $T_{\rm{train.}}$ out of the $T_{\rm{stat.}}$ blocks are used in covariance estimation, $(1-\frac{T_{\rm{train.}}}{T_{\rm{stat.}}})$ is the fraction of blocks left for data transmission. With this, the effective achievable rate is estimated as~\cite{ElAyach2014Spatially,Alkhateeb2016Frequency}
\begin{align}
R&\!=\!\frac{(1-\frac{T_{\rm{train.}}}{T_{\rm{stat.}}})}{K}\sum_{k=1}^{K}\log_2\Big | \bI_{\Ns}\!+\!\frac{P}{K\Ns} \bsfR_{\bsfn}[k]^{-1}\bWBB^\ast[k]\bWRF^\ast\times\nonumber\\
&~~~~~~~~~~~~\bsfH[k]\bFRF\bFBB[k]\bFBB^\ast[k]\bFRF^\ast\bsfH^\ast[k] \bWRF\bWBB[k]\Big |,
\end{align}
where $\bsfR_{\bsfn}[k]=\sigma_{\bsfn}^2\bWBB[k]^\ast\bWRF\bWRF\bWBB[k]$ is the noise covariance matrix after combining.

The \subsGHz~system operates at $3.5\GHz$ with $150\MHz$ bandwidth, and $\ulNRX=8$, and $\ulNTX=4$ antennas. The mmWave system operates at $28~\GHz$ with $850\MHz$ bandwidth. The bandwidths are the maximum available bandwidths in the respective bands~\cite{3p5GHz,28GHz}. The mmWave system has $\NRX=64$, and $\NTX=32$ antennas, and $\MRX=16$, and $\MTX=8$ RF-chains. The number of transmitted streams is $\Ns=4$. Both systems use ULAs with half wavelength spacing $\ulDelta=\Delta=1/2$. With the chosen frequencies, the number of antennas, and inter-element spacing at \subsGHz~and mmWave, the \subsGHz~and mmWave antennas arrays have the same aperture. The transmission power for \subsGHz~is $\ul{P}=30~\dBm$ per $25~\MHz$~\cite{Commission2012Amendment}, and for mmWave is $P=43~\dBm$~\cite{Commission2016Fact}. The path-loss coefficient at \subsGHz~and mmWave is $3$ and the complex path coefficients of the channels are IID complex Normal.  The number of \subsGHz~OFDM sub-carriers is $\ul{K}=32$ and mmWave OFDM sub-carriers is $K=128$. The CP length is a quarter of the symbol duration for both \subsGHz~and mmWave. The number of delay-taps in \subsGHz~and mmWave is one more than the length of CP, i.e.,  $\ul{D}=9$ and $D=33$ delay-taps. A raised cosine filter with a roll-off factor of $1$ is used for pulse shaping. The MDL algorithm~\cite{Wax1985Detection} is used to estimate the number of clusters for covariance translation. A $2$x over-complete DFT basis. i.e., $B_{\rmR\rmX}=2\NRX$ and $B_{\rmT\rmX}=2\NTX$ is used for compressed covariance estimation. Two-bit phase-shifters based analog precoders/combiners are used for simulations are based on.

We start by considering a simple two-cluster channel, where each cluster contributes $100$ rays.  We assume that all the rays within a cluster arrive at the same time. We use Gaussian PAS with $\SI{3}{\degree}$ AS. To calculate the efficiency metric~\eqref{eq:effmtc}, we use the theoretical expressions of the covariance matrix with Gaussian PAS (see Table~\ref{tab:ThExp}) in~\eqref{eq:totalcovmatrixmmWave} as the true covariance. The TX-RX distance is $\SI{90}{\metre}$ and the number of snapshots for \subsGHz~covariance estimation is $30$. 

We present the results of covariance translation as a function of the separation between mean AoA of the clusters. Specifically, the mean AoA of one cluster is fixed at $\SI{5}{\degree}$ and the mean AoA of the second cluster is varied from $\SI{5}{\degree}$-$\SI{20}{\degree}$. The difference between the mean AoAs of the clusters is the separation in degrees. We assume that power contribution of the clusters is the same i.e., $\epsilon_1=\epsilon_2=0.5$. The time of arrival of the cluster at $\SI{5}{\degree}$ is fixed at $0$. For the other cluster, the time of arrival is chosen uniformly at random between $0$ to $\SI{10}{\nano\second}$. We plot the number of clusters estimated in the proposed translation strategy vs mean AoA separation in Fig.~\ref{fig:hatCSeparation}. Note that due to the robustification discussed earlier, it is possible that the final number of estimated clusters be different than the estimate provided by MDL. We are plotting the final number of estimated clusters. For small separations, effectively the channel has a single-cluster, and hence a single-cluster is estimated. As the separation increased, the algorithm can detect one or two clusters. With large enough separation, the algorithm successfully determines two clusters. The Fig.~\ref{fig:etaSeparation} shows the efficiency metric of the proposed strategy versus separation. When separation is below  $\SI{8}{\degree}$, and two clusters are detected, their AoA and AS estimation is erroneous due to small separation and the efficiency is low. As the separation increases, the AoA and AS estimation improves and with it the efficiency of the covariance translation approach.

\begin{figure}
\centering
\includegraphics[width=0.45\textwidth]{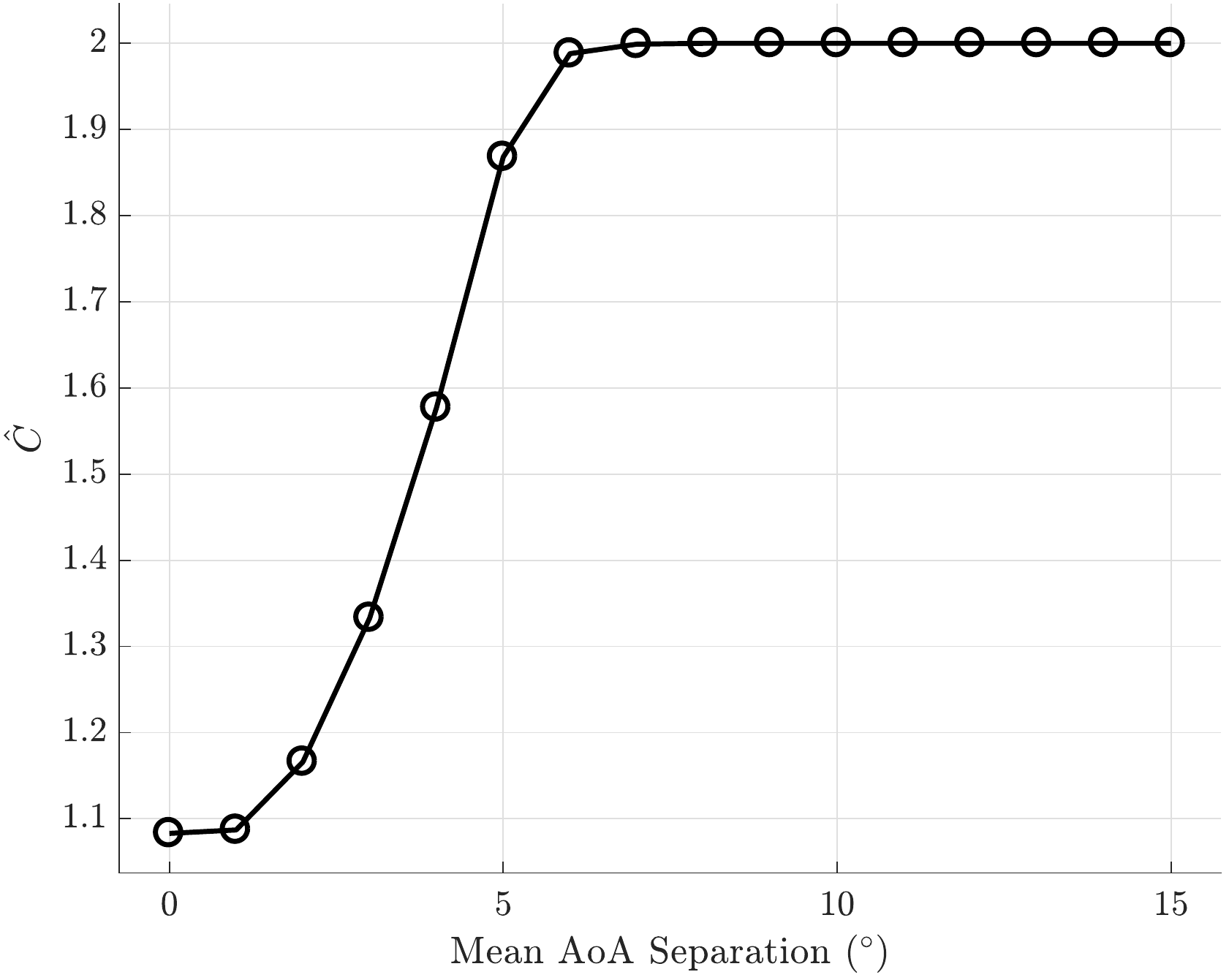}
\caption{The estimated number of clusters $\hat C$ (in a two-cluster channel) versus the mean AoA separation ($\SI{}{\degree}$) of the proposed covariance translation strategy. The mean AoA of the first cluster is $5^\circ$ and the mean AoA of the second cluster is varied from $5^\circ$ to $20^\circ$. The TX-RX distance is $\SI{90}{m}$.}
\label{fig:hatCSeparation}
\end{figure}

\begin{figure}
\centering
\includegraphics[width=0.45\textwidth]{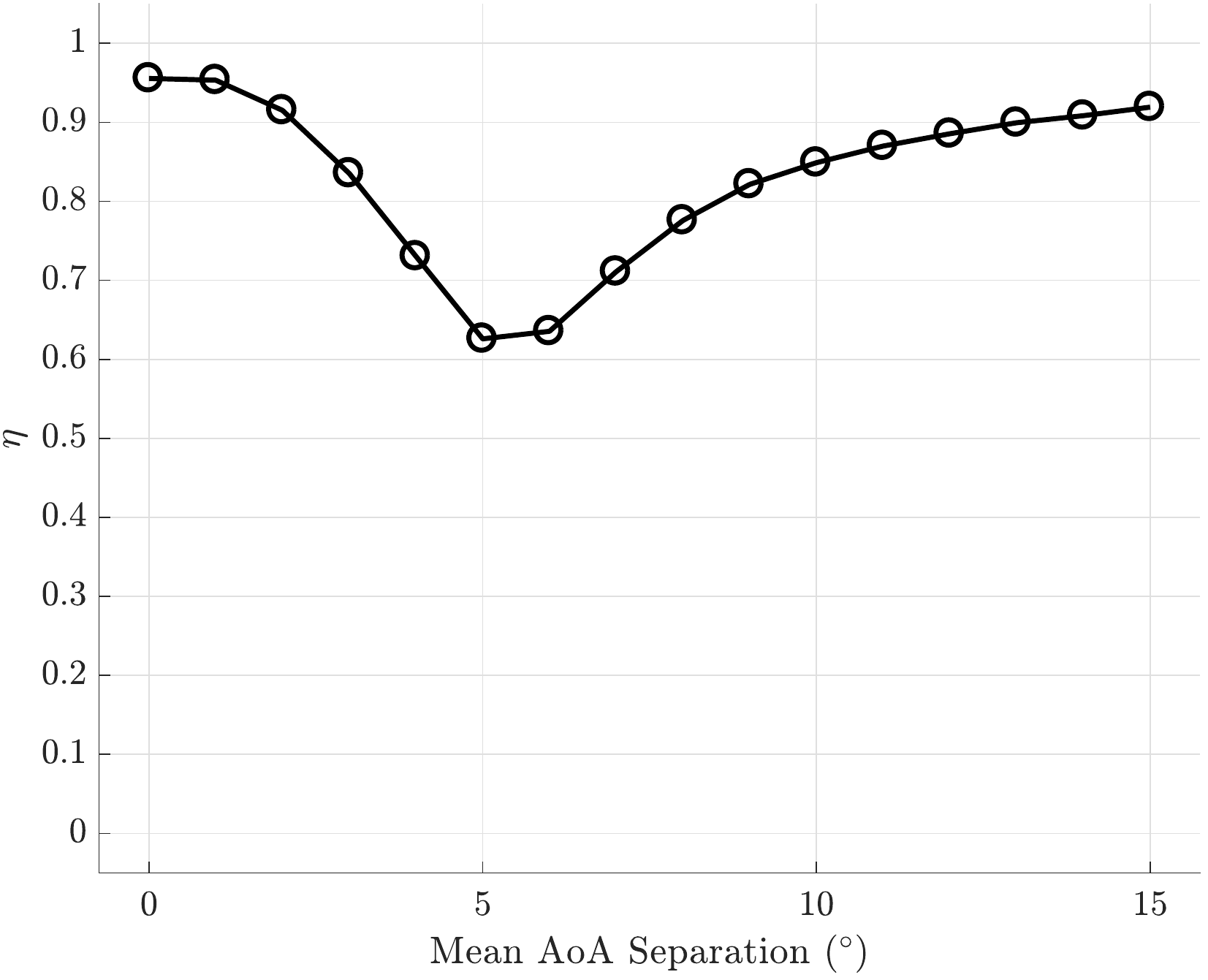}
\caption{The efficiency metric $\eta$ (in a two-cluster channel) versus the mean AoA separation ($\SI{}{\degree}$) of the proposed covariance translation strategy. The mean AoA of the first cluster is $5^\circ$ and the mean AoA of the second cluster is varied from $5^\circ$ to $20^\circ$. The TX-RX distance is $\SI{90}{m}$.}
\label{fig:etaSeparation}
\end{figure}

We test the performance of the proposed LW-DCOMP algorithm as a function of TX-RX distance. The number of snapshots is $30$. We fix the clusters at $5^\circ$ and $45^\circ$, and the cluster powers at $\epsilon_1=\epsilon_2=0.5$. The time of arrival of the cluster at $\SI{5}{\degree}$ is fixed at $0$, and the time of arrival of the cluster at $\SI{45}{\degree}$ is chosen uniformly at random between $0$ to $\SI{10}{\nano\second}$. The results of this experiment are shown in Fig.~\ref{fig:etaTXRXdist}. We see that as the TX-RX separation increases - and the \SNR~decreases - the benefit of using out-of-band information in compressed covariance estimation becomes clear. 

\begin{figure}
\centering
\includegraphics[width=0.45\textwidth]{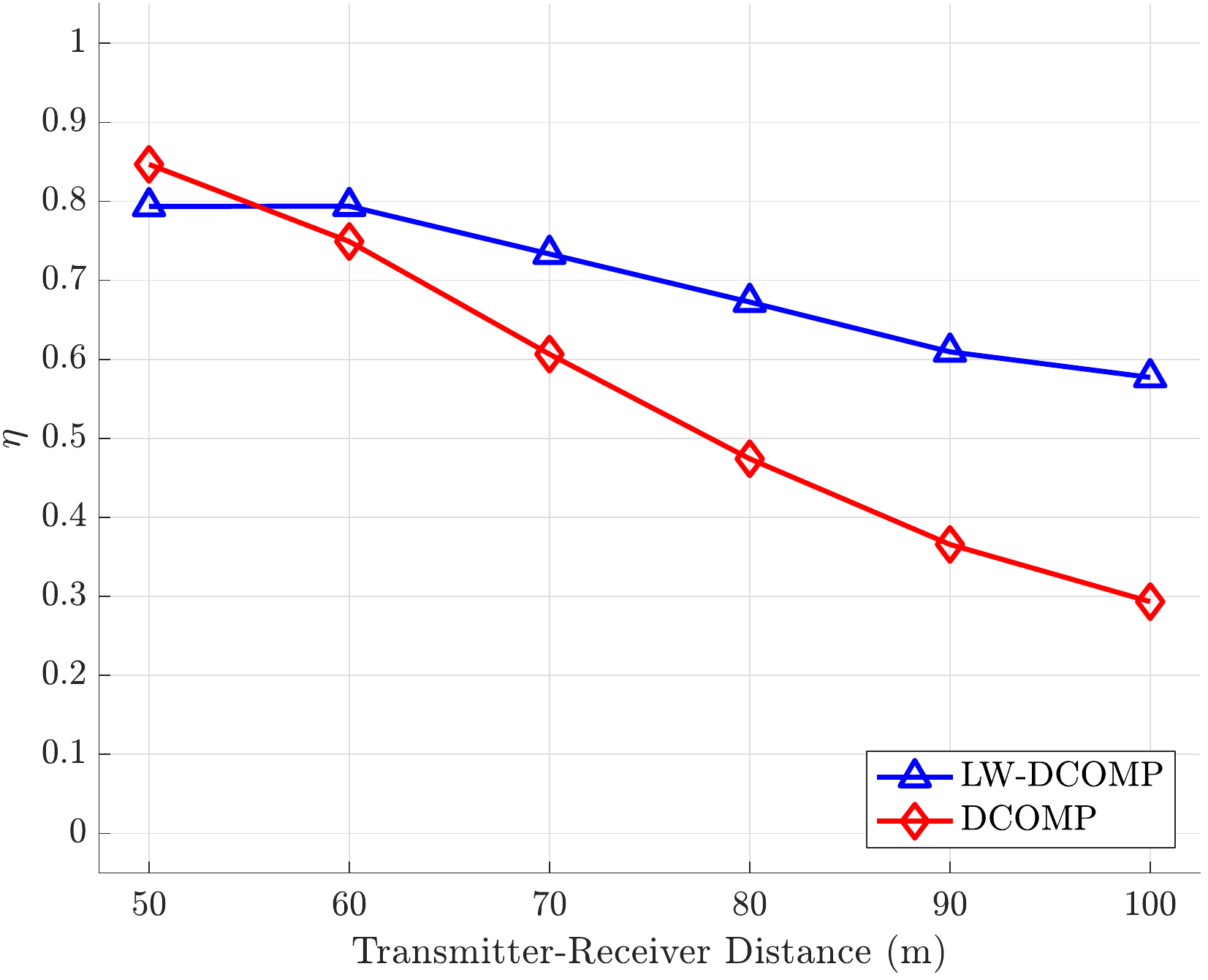}
\caption{The efficiency metric $\eta$ (in a two-cluster channel) of the the proposed LW-DCOMP algorithm versus the TX-RX distance ($\SI{}{\metre}$). First cluster has mean AoA $5^\circ$ and the second cluster has a mean AoA $45^\circ$. The number of snapshots $T$ is $30$. }
\label{fig:etaTXRXdist}
\end{figure}

So far we did not consider the spatial discrepancy in the \subsGHz~and mmWave channels. We now test the performance of the proposed strategies in more realistic channels, i.e., where \subsGHz~and mmWave systems have a mismatch. Specifically, we generate the channels according to the methodology proposed in~\cite{Ali2018Millimeter}. We refer the interested reader to~\cite{Ali2018Millimeter} for the details of the method to generate \subsGHz~and mmWave channels. Here, we only give the channel parameters. The \subsGHz~and mmWave channels have $\ul{C}=10$ and $C=5$ clusters respectively, each contributing $\ul{R}_{\rmc}=R_{\rmc}=20$ rays. The mean angles of the clusters are limited to $[-\frac{\pi}{3},\frac{\pi}{3})$. The relative angle shifts come from a wrapped Gaussian distribution with with AS $\{\ul{\sigma}_{\vartheta_{c}},\ul{\sigma}_{\varphi_{c}}\}=4^\circ$ and $\{\sigma_{\vartheta_c},\sigma_{\varphi_c}\}=2^\circ$
As the delay spread of \subsGHz~channel is expected to be larger than the delay spread of mmWave~\cite{weiler2015simultaneous,Poon2003Indoor,Jaeckel20165G,Kaya201628}, we choose $\ul{\tau}_{\mathrm{RMS}}\approx 3.8~\ns$ and $\tau_{\mathrm{RMS}}\approx 2.7~\ns$. The relative time delays of the paths within the clusters are drawn from zero mean Normal distributions with RMS AS $\ul{\sigma}_{\ul{\tau}_{r_{\ul{c}}}}=\frac{\ul{\tau}_{\mathrm{RMS}}}{10}$ and $\sigma_{\tau_{r_c}}=\frac{\tau_{\mathrm{RMS}}}{10}$. The powers of the clusters are drawn from exponential distributions. Specifically, the exponential distribution with parameter $\mu$ is defined as $f(x|\mu)=\frac{1}{\mu}e^{-\frac{x}{\mu}}$. The parameter for \subsGHz~was chosen as $\ul{\mu}=0.2$ and for mmWave $\mu=0.1$. This implies that the power in late arriving multi-paths for mmWave will decline more rapidly than~\subsGHz. The system parameters are identical to the previously explained setup. The hybrid precoders/combiners are designed using the greedy algorithm given in the~\cite{Alkhateeb2013Hybrid}. The effective achievable rate results are shown in Fig.~\ref{fig:Res_Rate_dist}. As the \subsGHz~CSI is obtained for establishing the \subsGHz~link, it does not cause any training overhead for mmWave link establishment. Therefore, we do not consider training overhead for computing the effective rate of covariance translation approach. For compressed covariance estimation and weighted compressed covariance estimation, we assume that $T_{\rm{stat.}}=2048$. The number of training OFDM blocks is $T_{\rm{train.}}=2\times2\times T$. Here, $T$ is the number of snapshots. A factor of $2$ appears as we use $2$ OFDM blocks to create omni-directional transmission, i.e., one snapshot. Another factor of $2$ appears as the training is performed for the transmit and the receive covariance estimation. The observations about the benefit of using out-of-band information in mmWave covariance estimation also hold in this experiment. Note that, the achievable rate drops with increasing TX-RX distance due to decreasing \SNR. Further, this experiment validates the robustness of the designed covariance estimation strategies to the correlated channel taps case.

\begin{figure}
\centering
\includegraphics[width=0.45\textwidth]{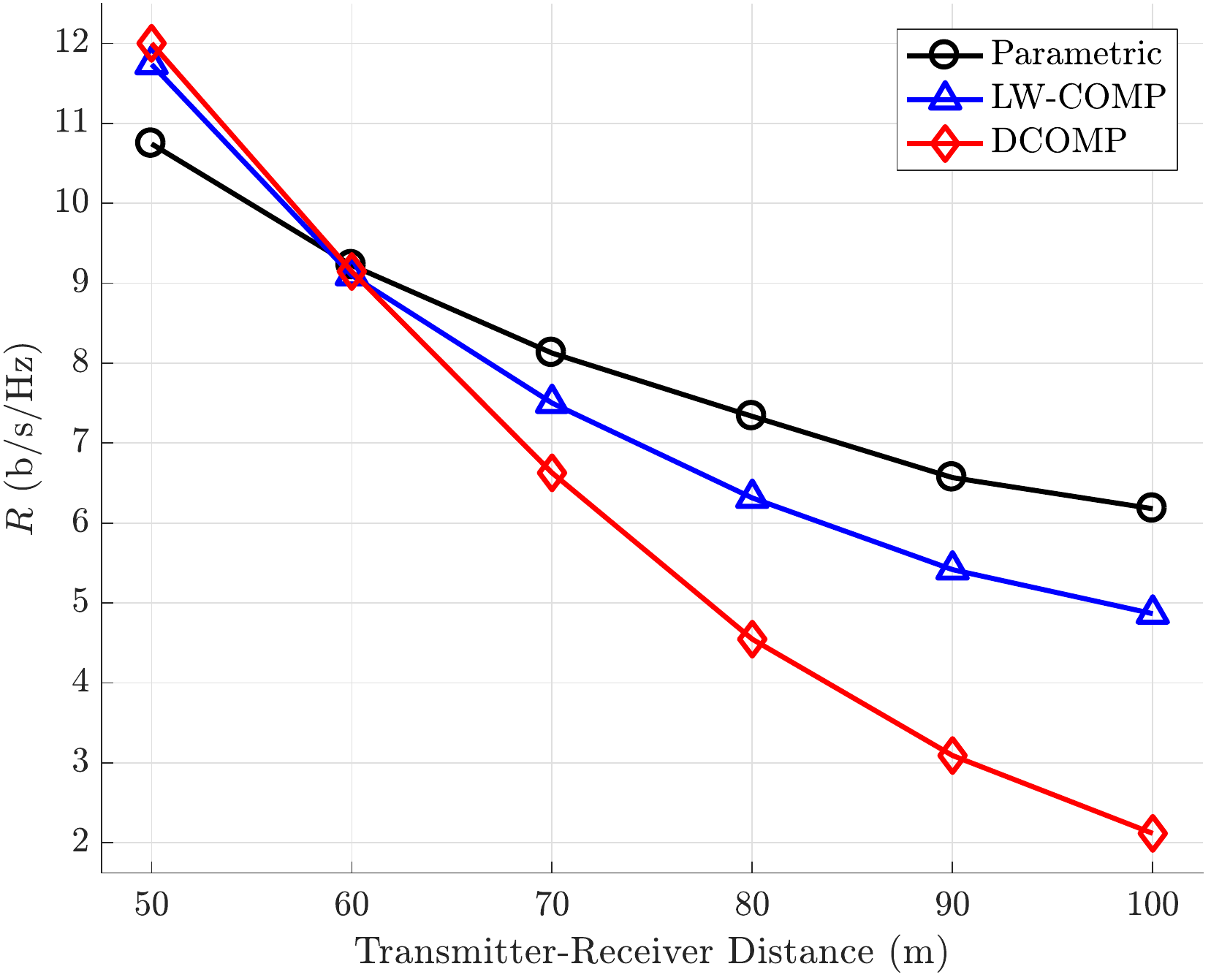}
\caption{The effective achievable rate of the proposed covariance estimation strategies versus the TX-RX distance. The rate calculations are based on $T_{\rm{stat.}}=2048$ blocks and $T_{\rm{train.}}=120$ blocks.}
\label{fig:Res_Rate_dist}
\end{figure}

We now compare the overhead of the proposed LW-DCOMP approach to the DCOMP approach. We use $T_{\rm{train.}}=2\times 2\times T$ for rate calculations. In Fig.~\ref{fig:Res_Rate_snapshots}, we plot the effective achievable rate versus the number of snapshots $T$ for three different values of $T_{\rm stat.}$. For dynamic channels, i.e., with $T_{\rm{stat.}}=1024$ or $T_{\rm{stat.}}=2048$, the effective rate of both LW-DCOMP and DCOMP increases with snapshots, but as we keep on increasing $T$, the rate starts to decrease. This is because, though the channel estimation quality increases, a significant fraction of the $T_{\rm{stat.}}$ is spent training and the thus there is less time to use the channel for data-transmission. Taking $T_{\rm{stat.}}=2048$ as an example, the highest rate of DCOMP algorithm is $7.16~\rm{b/s/Hz}$ and is achieved with $45$ snapshots. In comparison, the optimal rate of the LW-DCOMP algorithm is $7.46~\rm{b/s/Hz}$ and is achieved with only $25$ snapshots. The LW-DCOMP achieves a rate better than the highest rate of DCOMP algorithm ($7.16~\rm{b/s/Hz}$) with less than $15$ snapshots. Thus, the LW-DCOMP can reduce the training overhead of DCOMP by over $3$x. 
\begin{figure}
\centering
\includegraphics[width=0.45\textwidth]{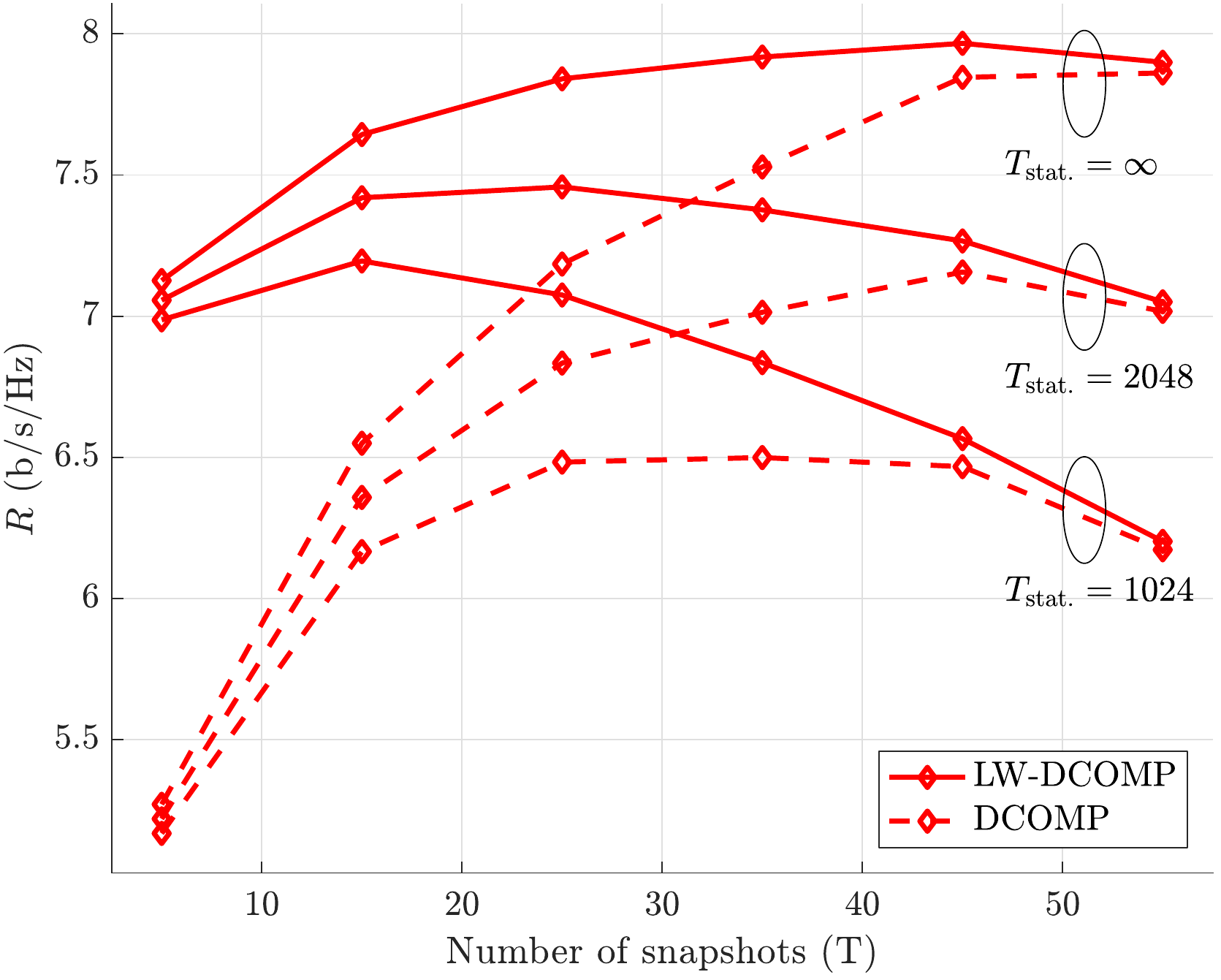}
\caption{The effective achievable rate of LW-DCOMP and DCOMP versus the number of snapshots $T$ (at the transmitter and the receiver). The effective rate is plotted for three values of $T_{\rm{stat.}}$. The TX-RX distance is fixed at $\SI{70}{\meter}$.}
\label{fig:Res_Rate_snapshots}
\end{figure}


Now we verify the \SNR~loss analysis outlined in Section~\ref{sec:analysis}. For this purpose, we consider two mmWave systems with $\NTX=\NRX=64$ and $\NTX=\NRX=16$ antennas. The number of RF chains in both cases is $\MTX=\MRX=\sqrt{\NTX}=\sqrt{\NRX}$. We plot the loss in $\SNR$ $\gamma$ as a function of the per-subcarrier \SNR~i.e., $\SNR_k=\frac{P}{\sigma_{\bsfn}^2 K}$. We assume that $[\delbsfRRX]_{i,j}=[\delbsfRTX]_{i,j}\sim\cC\cN(0,\frac{1}{{\SNR_k}})$. The smallest singular value of the Gaussian matrices vanishes as the dimensions increase~\cite{Vershynin2010Introduction}, and the lower bound becomes trivial. As such, we only show the results for the upper bound in Fig.~\ref{fig:Res_eta_SNR}. We show the upper bound as predicted by the analysis and empirical difference in the average $\SNR$ of the mmWave systems based on true covariance and the perturbed covariance. The empirical difference is plotted for the case when the singular vectors are used as precoder/combiner (i.e., assuming fully digital precoding/combining) and also for the case when hybrid precoders/combiners are used. From the results, we can see that the upper bound is valid for both systems with $16$ and $64$ antennas respectively. An interesting observation is that when the hybrid precoders/combiners are used in the mmWave system, the loss due to the mismatch in the estimated and true covariance is less than the case when fully digital precoding and combining is used.

\begin{figure}
\centering
\includegraphics[width=0.45\textwidth]{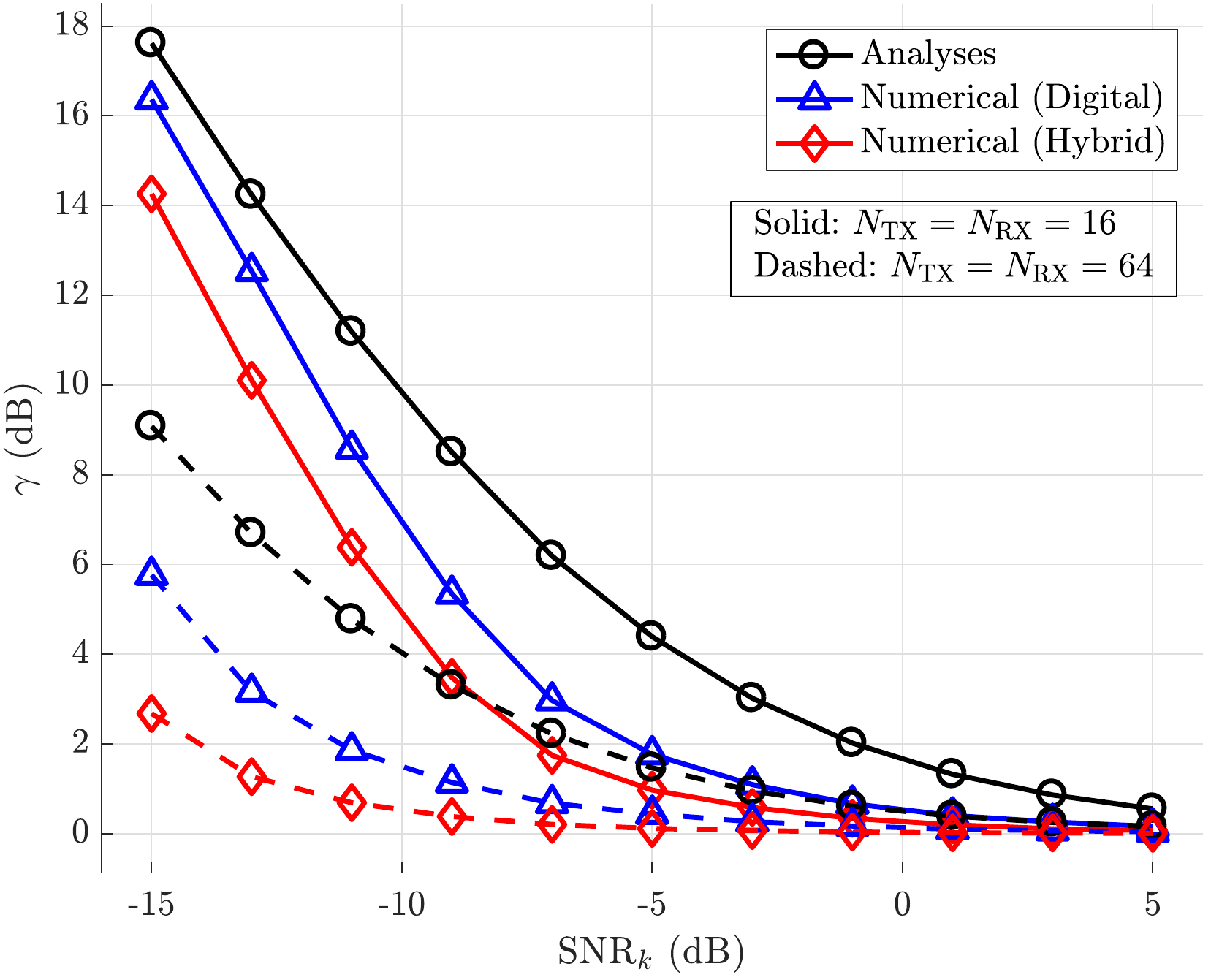}
\caption{The upper bound on the $\SNR$ loss $\gamma$ with complex Normal perturbation.    The number of RF chains is $\MTX=\MRX=\sqrt{\NTX}=\sqrt{\NRX}$ and $\SNR_k=\frac{P}{\sigma_{\bsfn}^2 K}$.}
\label{fig:Res_eta_SNR}
\end{figure}
\section{Conclusion}\label{sec:conc}
In this paper, we used the \subsGHz~covariance to predict the mmWave~covariance. We presented a parametric approach that relies on the estimates of mean angle and angle spread and their subsequent use in theoretical expressions of the covariance pertaining to a postulated power azimuth spectrum. To aid the in-band compressed covariance estimation with out-of-band information, we formulated the compressed covariance estimation problem as weighted compressed covariance estimation. For a single path channel, we bounded the loss in \SNR~caused by imperfect covariance estimation using singular-vector perturbation theory. 

The out-of-band covariance translation and out-of-band aided compressed covariance estimation had better effective achievable rate than in-band only training, especially in low \SNR~scenarios. The out-of-band covariance translation eliminated the in-band training but performed poorly when the \SNR~was favorable. The out-of-band aided compressed covariance estimation reduced the training overhead of the in-band only covariance estimation by $3$x. As a future work, covariance estimation approaches for array geometries other than uniform linear arrays, e.g., circular and planar arrays should be explored.
\bibliographystyle{IEEEtran}
\section*{Appendix A\\Proof of Theorem 1}
The received signal~\eqref{eq:rxdestcov} can be written as
\begin{align}
\sfy&=\frac{\bURXs^\ast \bsfH \bUTXs}{\|\hbURXs\|\|\hbUTXs\|} \sfs+\frac{\delURXs^\ast \bsfH \bUTXs}{\|\hbURXs\|\|\hbUTXs\|}\sfs\nonumber\\
&+\frac{\bURXs^\ast \bsfH \delUTXs}{\|\hbURXs\|\|\hbUTXs\|}\sfs+\frac{\delURXs^\ast \bsfH \delUTXs}{\|\hbURXs\|\|\hbUTXs\|}\sfs+\frac{\hbURXs^\ast}{\|\hbURXs\|}\bsfn.
\label{eq:rxdexpanded}
\end{align}
The numerator of the first term on the RHS is identical to the first term in~\eqref{eq:rxdtruecov} and can be simplified as
\begin{align}
\frac{\bURXs^\ast \bH \bUTXs}{\|\hbURXs\|\|\hbUTXs\|}\sfs=\frac{\sqrt{\NRX\NTX}\alpha \sfs}{\|\hbURXs\|\|\hbUTXs\|}.
\end{align}
Now using the channel representation in form of the signal subspace $\bsfH=\sqrt{\NRX\NTX}\alpha\bURXs \bURXs^\ast$, the second term can be written as
\begin{align}
\frac{\delURXs^\ast \bsfH \bUTXs}{\|\hbURXs\|\|\hbUTXs\|}\sfs=\frac{\sqrt{\NRX\NTX}\alpha\delURXs^\ast \bURXs \bURXs^\ast \bUTXs}{\|\hbURXs\|\|\hbUTXs\|}\sfs.
\end{align}
Using the results from~\cite{Liu2008First} we can write $\delURXs$ as
\begin{align}
\delURXs=\frac{\bURXn\bURXn^\ast \delbsfRRX \bURXs}{\sigma_\alpha^2\NRX}.
\label{eq:delURXsdef}
\end{align}
and further
\begin{align}
\delURXs^\ast\bURXs=\frac{\bURXs^\ast \delbsfRRX^\ast \bURXn \bURXn^\ast \bURXs}{\sigma_\alpha^2\NRX}.
\label{eq:delurxnurxs}
\end{align}
In~\eqref{eq:delurxnurxs}, $\bURXn^\ast \bURXs=\bzero$, and hence the second term in~\eqref{eq:rxdexpanded} is zero. The third term and the fourth term in~\eqref{eq:rxdexpanded} vanish by the same argument. Hence the received signal can be simply written as
\begin{align}
\sfy=\frac{\sqrt{\NRX\NTX}\alpha \sfs}{\|\hbURXs\|\|\hbUTXs\|}+\frac{\hbURXs^\ast}{\|\hbURXs\|}\bsfn,
\end{align}
from which \SNR~expression~\eqref{eq:SNRest} can be obtained.

\section*{Appendix B\\Proof of Theorem 2}
We work solely on simplifying $\hbURXs$ on the RHS of~\eqref{eq:snrlosssigsub} as the simplification of $\hbUTXs$ is analogous. We can write
\begin{align}
\|\hbURXs\|^2&\overset{(a)}{=}\|\bURXs+\delURXs\|^2, \nonumber\\
&\overset{(b)}{=}\|\bURXs\|^2+\|\delURXs\|^2,\nonumber\\
&\overset{(c)}{=}1+\frac{1}{\sigma_\alpha^4\NRX^2}\left\|\bURXn\bURXn^\ast \delbsfRRX \bURXs\right \|^2,\nonumber\\
&\overset{(d)}{\approx}1+\frac{1}{\sigma_\alpha^4\NRX^2}\left\|\delbsfRRX \bURXs\right \|^2,
\end{align}
where $(a)$ comes from the definition of $\hbURXs$ and $(b)$ comes from the assumption that the phase of the perturbation is adjusted to have the true signal subspace and the perturbation signal subspace orthogonal~\cite{Liu2008First}. In $(c)$ the first term simplifies to $1$ as the norm of the singular vector and the second term comes from the definition of $\delURXs$ in~\eqref{eq:delURXsdef}. Finally in $(d)$ we use the approximation $\bURXn\bURXn^\ast\approx \bI$. Note that for single path channels, the signal space is one dimensional and $\frac{\|\bURXn\bURXn^\ast-\bI_{\NRX}\|_\rmF^2}{\|\bI_{\NRX}\|_\rmF^2}=\frac{1}{\NRX}$. Hence $\frac{\|\bURXn\bURXn^\ast-\bI_{\NRX}\|_\rmF^2}{\|\bI_{\NRX}\|_\rmF^2}\rightarrow 0$ as $\NRX \rightarrow \infty$ and the approximation is exact in the limit. For mmWave systems where the number of antennas is typically large, the approximation is fair. Following an analogous derivation for $\hbUTXs$, we get~\eqref{eq:snrlossapprox}.

To obtain the upper and lower bounds, note the following about the norm of a matrix-vector product $\bA\bx$: 
$\max_{\|\bx\|_2=1} \|\bA\bx\|^2 = \sigma_{\max}^2(\bA)$ and $\min_{\|\bx\|_2=1} \|\bA\bx\|^2 = \sigma_{\min}^2(\bA)$,
where $\sigma_{\max}(\bA)$ and $\sigma_{\min}(\bA)$ is the largest and smallest singular value of the matrix $\bA$. As in~\eqref{eq:snrlossapprox}, $\bURXs$ is a singular vector with unit norm, we can bound $\sigma_{\min}(\delbsfRRX) \leq \|\delbsfRRX \bURXs\| \leq\sigma_{\max}(\delbsfRRX)$. Using this result in~\eqref{eq:snrlossapprox}, we get~\eqref{eq:snrlossupper} and~\eqref{eq:snrlosslower}.
\bibliography{\centrallocation/Abbr,\centrallocation/Master_Bibliography}{}

\end{document}

%% file: macros.tex
\usepackage{amsmath,amssymb}

\newcommand{\bbC}{{\mathbb{C}}}

\newcommand{\bbE}{{\mathbb{E}}}

\newcommand{\bbR}{{\mathbb{R}}}

\newcommand{\ba}{{\mathbf{a}}}

\newcommand{\bff}{{\mathbf{f}}}
\newcommand{\bg}{{\mathbf{g}}}

\newcommand{\bp}{{\mathbf{p}}}

\newcommand{\bx}{{\mathbf{x}}}
\newcommand{\by}{{\mathbf{y}}}

\newcommand{\bzero}{{\mathbf{0}}}
\newcommand{\bone}{{\mathbf{1}}}
\newcommand{\bA}{{\mathbf{A}}}

\newcommand{\bF}{{\mathbf{F}}}

\newcommand{\bH}{{\mathbf{H}}}
\newcommand{\bI}{{\mathbf{I}}}

\newcommand{\bR}{{\mathbf{R}}}

\newcommand{\bU}{{\mathbf{U}}}
\newcommand{\bV}{{\mathbf{V}}}
\newcommand{\bW}{{\mathbf{W}}}
\newcommand{\bX}{{\mathbf{X}}}

\newcommand{\rmc}{{\mathrm{c}}}

\newcommand{\rmm}{{\mathrm{m}}}
\newcommand{\rmn}{{\mathrm{n}}}

\newcommand{\rms}{{\mathrm{s}}}

\newcommand{\rmw}{{\mathrm{w}}}

\newcommand{\rmB}{{\mathrm{B}}}

\newcommand{\rmF}{{\mathrm{F}}}

\newcommand{\rmP}{{\mathrm{P}}}

\newcommand{\rmR}{{\mathrm{R}}}
\newcommand{\rmS}{{\mathrm{S}}}
\newcommand{\rmT}{{\mathrm{T}}}

\newcommand{\rmX}{{\mathrm{X}}}

\newcommand{\cC}{\mathcal{C}}

\newcommand{\cN}{\mathcal{N}}

\newcommand{\cS}{\mathcal{S}}

\newcommand{\sfs}{{\mathsf{s}}}

\newcommand{\sfy}{{\mathsf{y}}}

\newcommand{\bsfg}{\boldsymbol{\mathsf{g}}}

\newcommand{\bsfn}{\boldsymbol{\mathsf{n}}}

\newcommand{\bsfr}{\boldsymbol{\mathsf{r}}}
\newcommand{\bsfs}{\boldsymbol{\mathsf{s}}}

\newcommand{\bsfx}{\boldsymbol{\mathsf{x}}}
\newcommand{\bsfy}{\boldsymbol{\mathsf{y}}}

\newcommand{\bsfF}{\boldsymbol{\mathsf{F}}}

\newcommand{\bsfH}{\boldsymbol{\mathsf{H}}}

\newcommand{\bsfR}{\boldsymbol{\mathsf{R}}}

\newcommand{\bsfW}{\boldsymbol{\mathsf{W}}}

\newcommand{\balpha}{\boldsymbol{\alpha}}

\newcommand{\brho}{\boldsymbol{\rho}}

\newcommand{\bSigma}{\boldsymbol{\Sigma}}

\newcommand{\transp}{{\sf T}}
\newcommand{\compj}{{\rm j}}
\newcommand{\tr}{{\rm tr}}
\renewcommand{\vec}{{\rm vec}}
\newcommand{\diag}{{\rm diag}}

\makeatletter
\def\munderbar#1{\underline{\sbox\tw@{$#1$}\dp\tw@\z@\box\tw@}}
\makeatother